\documentclass[12pt]{article}

\usepackage{enumerate}
\usepackage{mathtools}
\mathtoolsset{showonlyrefs=true}
\usepackage{todonotes}

\usepackage{url}
\usepackage{amsfonts}
\usepackage{amsmath,amssymb,color}
\usepackage{graphicx}
\usepackage{epstopdf}
\usepackage{multirow}
\usepackage{bm}
\usepackage{natbib}
\usepackage{verbatim}
\usepackage{float}
\usepackage[toc,page]{appendix}
\usepackage[onehalfspacing]{setspace}
\usepackage{setspace}
\usepackage{adjustbox}
\usepackage{rotating}
\usepackage{subcaption}
\usepackage{amsthm}
\usepackage{booktabs,dcolumn,caption}

\usepackage{mathabx}

\usepackage{array}
\usepackage{algpseudocode,algorithm,algorithmicx}

\let\hat\widehat

\newcolumntype{L}[1]{>{\raggedright\let\newline\\\arraybackslash\hspace{0pt}}p{#1}}
\newcolumntype{C}[1]{>{\centering\let\newline\\\arraybackslash\hspace{0pt}}p{#1}}
\newcolumntype{R}[1]{>{\raggedleft\let\newline\\\arraybackslash\hspace{0pt}}p{#1}}

\newcommand{\yc}[1]{{{\color{black}  #1}}}

\newcommand{\PP}{{\mathbf{P}}}

\newcommand{\bbb}{{\boldsymbol \beta}}

\newcommand{\HH}{\mathbf{H}}
\newcommand{\hh}{\mathbf{h}}

\newcommand{\yy}{\mbox{$\mathbf y$}}

\newcommand{\ww}{\mathbf w}

\newcommand{\Zp}{\mathbb{Z}_+}

\newcommand{\card}{\dim}

\newcommand{\fd}{\mathfrak d}
\newcommand{\spaceo}{\mathcal{S}_{\mathrm{o}}}
\newcommand{\lleq}{\preccurlyeq}
\newcommand{\vv}{\mbox{$\mathbf v$}}
\newcommand{\ab}{\bm \varnothing}
\newcommand{\uu}{\mbox{$\mathbf u$}}
\newcommand{\eqd}{{\stackrel{d}{=}}}

\newcommand{\KKa}{\mathbb{K}_{\textrm{a}}}

\newcommand{\argmax}{\operatornamewithlimits{arg\,max}}

\newtheorem{theorem}{Theorem}

\newtheorem{lemma}{Lemma}

\newtheorem{proposition}{Proposition}
\newtheorem{remark}{Remark}

\title{Item Quality Control in Educational Testing: Change Point Model, Compound Risk, and Sequential Detection}
\author{Yunxiao Chen\\
London School of Economics and Political Science\\
Yi-Hsuan Lee\\
Educational Testing Service \\
Xiaoou Li\\
University of Minnesota}
\date{}

\begin{document}
\maketitle

\doublespacing

\begin{abstract}
In standardized educational testing, test items are reused in multiple test administrations. To ensure the validity of test scores, the psychometric properties of items should remain unchanged over time. In this paper, we consider the sequential monitoring of test items,
in particular, the detection of abrupt changes to their psychometric properties, where a change can be caused by, for example, leakage of the item or change of the corresponding curriculum.
We propose a statistical framework for
the detection of abrupt changes in individual items.
This framework consists of (1) a multi-stream Bayesian change point model describing sequential changes
in items, (2) a compound risk function quantifying the risk in sequential decisions, and (3) sequential decision rules that control the compound risk. Throughout the sequential decision process, the proposed decision rule balances the trade-off between two sources of errors, the
false detection of pre-change items and the non-detection of post-change items.
An item-specific monitoring statistic is proposed based on an item response theory model that eliminates the confounding from the examinee population which changes over time. Sequential decision rules and their theoretical properties are developed under two settings: the oracle setting where the Bayesian change point model is completely known and a more realistic setting where some parameters of the model are unknown.
Simulation studies are conducted under settings that mimic real operational tests.
\end{abstract}	

\section{Introduction}

The administration of a standardized educational test typically relies on an item pool, where items are repeatedly chosen from the pool to
assemble test forms. To maintain the validity and reliability of a standardized test over time, it is important to ensure that the psychometric properties of items in the pool remain unchanged. An item may need to be revised or removed from the pool once its psychometric properties
encounter a significant change, which may be caused by various reasons such as its leakage to the public or change of the corresponding curriculum. An important and challenging problem that test administrators face is to periodically review their testing data and detect the changed items as early as possible (see Chapter 4, \citealp{aera2014standards}).

{Following the discussion in \cite{lee2021studying}, we divide educational tests into two categories -- tests with infrequent  and frequent testing schedules. Infrequent tests include college admission tests like American College Test (ACT) and Scholastic Assessment Test (SAT) and survey assessments like the National Assessment of Educational Porgress (NAEP) and the Programme for International Student Assessment (PISA), where ACT and SAT deliver seven administrations per year, and survey assessments are typically delivered once every a few years. Frequent tests include both continuous tests that are delivered daily and frequent but non-continuous tests that are available once or more times per week. Examples of continuous tests include the Graduate Record Examinations (GRE) general test and the Praxis Core Academic Skills for Educators tests that follow the Multi-Stage Testing (MST, \citealp{yan2016computerized}) and  fixed-form testing designs, respectively. Frequent but non-continuous tests are also very common. For example, the Test of English for International Communication  (TOEIC) speaking test had 281 administrations in 3 years \citep{qu2017evaluating}, and another  assessment of English proficiency had 498 administrations in 6 years \citep{qian2021model}, both of which are fix-form tests.
This paper considers item pool monitoring for frequent tests, in which items are reused more frequently
and thus are more likely to be leaked. In particular, we focus on frequent but non-continuous tests with a fixed-form design. Generalization to other  frequent test settings is also discussed.}

To tackle this problem, we adopt a multi-stream sequential change point formulation. Each item corresponds to a data stream, for which data are collected sequentially from test administrations over time. Each data stream is associated with its own change point. The change point corresponds to a distributional change in some monitoring statistics which reflect certain psychometric properties of the item. That is, the monitoring statistics follow one distribution at any time before the change point, and follow a different distribution after the change point. Roughly speaking, our goal is to detect as many post-change items as possible at each time point, without making too many false detections of pre-change items. The detected items will be reviewed by the test developers to check their validity. Further actions, such as removing items, may follow.

To provide a sensible solution to this change detection problem for item quality control,
we propose a statistical decision framework. This framework consists of (1) a
multi-stream Bayesian change point model describing the data streams with change points,
(2) a compound risk function quantifying the risk in sequential decisions, and (3) sequential
decision rules that aim at detecting as many post-change items as possible while controlling the compound risk to be below a pre-specified level. Specifically, our risk function can be viewed as a measure of the proportion of post-change items among the undetected items given all the {up-to-date}  information {from the monitoring statistics, where the information will be formalized by an information filtration (see Section~\ref{subsec:general} for the definition)}. We emphasize that this risk function measures the overall item-pool-level risk, instead of
item-level risk.
It is thus suitable for the purpose of controlling item pool quality as a whole.
The quality of undetected items can be guaranteed by controlling their compound risk.
Consequently,
only the detected items need a validity check.
Our development considers two different settings, including an oracle setting for which the Bayesian change point model is completely known, and a more realistic setting where only partial information is available about the model.

The current development is a significant extension of \cite{chen2019compound}, where the compound sequential change detection
framework is first proposed. 
First, a more general model is considered 
that is more suitable for item quality control. Specifically, it takes into account
item exposure and addition of new items, two important features of test administration and maintenance.
Second, we extend the development in \cite{chen2019compound} to a more realistic setting when  only partial information is available about the Bayesian change point model. A change detection procedure is proposed that is shown to control the compound risk.
Third, a monitoring statistic is proposed based on an item response theory model. This statistic adjusts for confounding from examinee population changes, so that changes in item properties can be better detected.
Finally, simulation studies are conducted under settings that mimic the administrations of operational tests.

The proposed framework is closely related to, but substantially different from,  the classical sequential change detection problem for a single data stream  \citep{shiryaev1963optimum,roberts1966comparison,page1954continuous,shewhart1931economic}, as well as recent developments on change detection for multiple streams \citep[e.g.,][]{mei2010efficient,xie2013sequential,chan2017optimal,chen2015graph,chen2019sequential,chen2020false}.
The major difference is that the existing works, except for  \cite{chen2019compound}  and \cite{chen2020false}, consider the detection of a single change point, even with multi-stream data. Consequently, they do not handle a compound risk which aggregates information on the change points of different data streams.

This framework also closely connects to compound decision theory \citep[see e.g.,][]{zhang2003compound} which dates back to the seminal works of \cite{robbins1951asymptotically,robbins1956empirical}. Specifically, the compound risk that we control at each time point
can be viewed as a local false non-discovery rate studied in  \cite{efron2001empirical} and \cite{efron2004large,efron2008microarrays,efron2012large} for
testing multiple hypotheses.
The same risk measure has been applied in \cite{chen2019statistical} for the detection of leaked items and cheating examinees in a single test administration. The proposed method shares the same scalability as the local false discovery and non-discovery rates for multiple testing.
That is,  no matter how large the item pool is, it is always sensible to use the proposed procedure without changing the threshold for compound risk, while, on the other hand, error metrics like the familywise error rate are far less scalable.
In the sequential analysis literature, 
the idea of compound decision is rarely explored, except in \cite{song2019sequential} and \cite{bartroff2018multiple} where
compound decision theory for sequential multiple testing  is developed, and in \cite{chen2019compound} where the compound decision framework for multi-stream change detection is first proposed.

The sequential monitoring of test quality has also been an important topic in the field of educational testing.
For example,
 to monitor item quality,
 \cite{veerkamp2000detection} applied the CUSUM method \citep{page1954continuous} to sequentially detect changes in item difficulty. \yc{\cite{zhang2014sequential} and \cite{zhang2016monitoring} proposed a series of sequential statistical hypothesis tests for monitoring the item pool of a computerised adaptive testing system.}
\cite{choe2018sequential} proposed sequential change detection procedures for the detection of compromised items based on both item responses and response times.
\cite{lee2016test}
 used CUSUM statistics to monitor item performance and detect item preknowledge in continuous testing.
The existing methods focus on detecting changes in individual items, while, as we will discuss in Section~\ref{sec:dec}, the proposed compound decision framework provides better integrative decisions for the entire item pool.
To monitor general test quality, \cite{lee2013monitoring} proposed sequential procedures to monitor score stability and assess scale drift of an educational assessment over time.

The rest of the paper is organized as follows. In Section~\ref{sec:model}, we propose a general Bayesian change point model and a compound sequential detection procedure, followed by a specific model that can be easily implemented in operational tests.  The theoretical properties of the proposed decision rule are established. Section~\ref{sec:unknown}
extends the development to a more realistic setting where some parameters of the Bayesian change point model are unknown. A compound decision rule is proposed under this setting and its statistical properties are proved.
Section~\ref{sec:stat}  discusses the problem of item quality control and the confounding due to the examinee population change over time. A
statistic based on an item response theory model is proposed, where this confounding factor is controlled.
The performance of the proposed method is further evaluated in Section~\ref{sec:sim} via simulation studies. We conclude the paper with remarks in Section~\ref{sec:conc}.

\section{Bayesian Change Point Detection}\label{sec:model}

\subsection{General Framework}\label{subsec:general}
We start with a general statistical framework for change detection in multiple data streams. {For ease of exposition, we consider a frequent but non-continuous test with fixed forms. Example tests of this type are discussed in the introduction. As will be discussed in Section~\ref{sec:dec}, the proposed procedure can be generalized to other frequent tests, such as continuous tests with Computerized Adaptive Testing (CAT) or Multiple-Stage Testing (MST) designs.} We use $t = 1, 2, ... $ to record test administrations; for example, $t = 1$ denotes the first test administration. Let $S_t$ denote the item pool at time $t$ that contains the items available for the $t$th test administration. The item pool is allowed to change over time, due to (1) the deletion of problematic items (e.g., items detected to have changed) and (2) the addition of new items. 

For each item $k$, we monitor a certain statistic that reflects the psychometric properties of the item, denoted as $X_{kt}$, based on the data from all examinees in the $t$th test administration. This statistic may be univariate or multivariate, calculated based on data from the $t$th test administration.
The monitoring statistic needs to be constructed carefully to adjust for confounding due to the change of the examinee population (e.g., seasonal effect), in order to reflect the real changes in individual items. For example, it is not a good idea to simply monitor the item percent correct. This is because, the item percent correct may significantly increase in an test administration, due to that its examinee population overall has higher ability.
As will be discussed in Section~\ref{sec:stat},
one possible way to adjust for population change is by
using an item response theory model \citep[IRT;][]{lord2008statistical,lord1980applications}.
In addition, the value of $X_{kt}$ may be missing, as only a subset of items from the item pool $S_t$ will be used in the test administration.
We use
$S_t^* \subset S_t$ to denote the set of items being used in the $t$th test administration. That is, $X_{kt}$ is observed, if and only if $k\in S_t^*$. Figure~\ref{fig:flow} provides a flow chart illustrating this stochastic process. For the $t$th test administration, we start with an item pool $S_t$ which is determined by information from the previous test administrations. Then $S_t^*$ is selected from $S_t$ as the set of items for the $t$th test administration  by test developers or certain test assembly algorithms. For these items, response data are collected, leading to monitoring statistics $X_{kt}$, $k\in S_t^*$. Historical information will be used to determine the item pool for the $(t+1)$th administration, by deleting problematic items and/or adding new items.

\begin{figure}
  \centering
  \includegraphics[scale = 1.8]{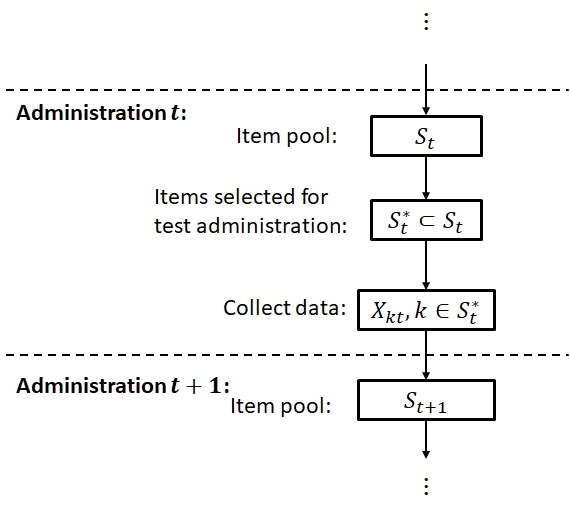}
  \caption{A flow chart for the stochastic process of sequential test administration.}\label{fig:flow}
\end{figure}

Each item $k$ is associated with a change point, denoted by $\tau_k$, which takes value in {$\{1, 2, ...\} \cup \{\infty\}$}.
More precisely, $\tau_k$ is the time at which the change  in the monitoring statistic $X_{kt}$ occurs to item $k$.
Here, we rule out the possibility that $\tau_k=0$ as it is sensible to assume that items can only change after some exposure\footnote{\yc{Theoretically, the proposed Bayesian change detection framework can allow $\tau_k=0$ to occur with a positive probability. We make the assumption that $\tau_k>0$ because the major application of the proposed method considers the detection of item compromisation that is due to previous exposure of the items \citep[see e.g.,][]{veerkamp2000detection}.}}.
For example, the change of an item may be due to
its leakage to the public at that time. The change time $\tau_k = \infty$ means that the item never changes.
\yc{Throughout this paper, we view $\tau_k$ as a random variable, whose prior distribution is allowed to vary across different streams.}
The distribution of $X_{kt}$ only depends on the change point $\tau_k$ and is independent of other variables in the model. That is,
\begin{equation}\label{eq:dists}
X_{kt}\vert \tau_k  \sim \left\{\begin{array}{cc}
                                 p_{kt} & \mbox{~if~} t \leq \tau_k \\
                                 q_{kt} & \mbox{~if~} t > \tau_k,
                               \end{array}\right.
\end{equation}
where $p_{kt}$ and $q_{kt}$ are the density functions of the pre- and post-change distributions, respectively. For now, we assume $p_{kt}$ and $q_{kt}$ are both known, for example, two normal distributions whose means and variances are given. We discuss in Section~\ref{sec:unknown} the situation when only partial information is available about these two distributions. Note that both the pre- and post-distributions may depend on $t$, for example, through the number of examinees in the corresponding test administration.
\yc{Once data have been collected at time $t$, we would like to detect the post-change items, given all the information that is currently available. The sequential detection rule can be described by a detection set $D_t \subset S_t$, consisting of items that are likely to have changed. We point out that the proposed framework is very general,
allowing
the detected items in $D_t$ to be removed or kept in the item pool $S_{t+1}$ at the next time point.}


Figure~\ref{fig:flow2} provides a toy example to illustrate this change point model. In this example, at $t = 1$, the item pool only contains items 1 and 2, and both are used in this test administration.
As no change has occurred to these two items yet, the monitoring statistics for both items follow their pre-change distributions as indicated by the circles. After the first test administration, a change point occurs to item 1, recorded by $\tau_1 = 1$.
At $t = 2$, item 3 is added to the item pool, leading to $S_2 = \{1, 2, 3\}$. In this test administration, items 1 and 3 are used and thus $S_2^* = \{1,3\}$.
As item 1 has already changed, the monitoring statistic $X_{12}$ now follows a post-change distribution, as indicated by the square.
Based on information from the first two test administrations, item 1 is detected and removed from the pool in this toy example, resulting in $S_3=\{2,3\}$.
The process can further evolve at $t = 4, 5, ...$, following the flow chart in Figure~\ref{fig:flow}. 


\begin{figure}
  \centering
  \includegraphics[scale = 1.8]{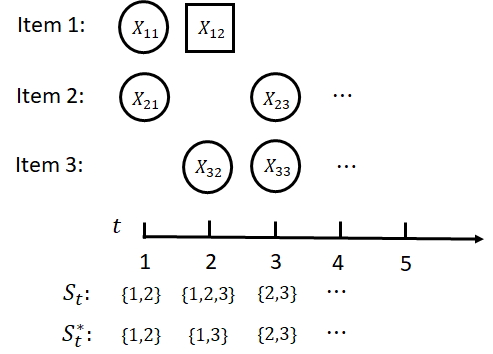}
  \caption{A toy example illustrating the stochastic process of test administration. The pre- and post-change distributions of a monitoring statistic are indicated by circle and square, respectively. }\label{fig:flow2}
\end{figure}

We detect changed items by monitoring existing data at each time point $t$. More precisely, the existing information after the $t$th test administration is recorded by a sigma-field $\mathcal F_t$ defined recursively as
$\mathcal F_t = \sigma(\mathcal F_{t-1}, S_t, S_t^*, X_{kt}, k\in S_t^*)$
with {$\mathcal F_1 = \sigma(\{S_1, S_1^*, X_{k1}, k\in S_1^*\})$.} {Note that the sigma-field $\mathcal F_t$ defines a information filtration satisfying $\mathcal F_s\subset \mathcal F_t$ for all $s\leq t$, meaning that as time passes, more and more information becomes available about the stochastic process of $X_{kt}$s.
Note that information filtration is a key concept in stochastic processes.
We refer the readers to \cite{florescu2014probability} for the mathematical details of an information filtration and its interpretation as accumulated information up to each time point.}

We consider a Bayesian setting where the change points $\tau_k$ are viewed as random variables, for which the posterior probabilities $P(\tau_k< t \vert \mathcal F_t)$ can be evaluated at each time point $t$ for any $k \in S_t$.
The detection of post-change items will be based on these posterior probabilities. \yc{More precisely, we consider sequential detection rules $D_t$ that are adaptive to the information filtration $\mathcal F_t$; i.e., $D_t$ is a random set that is measurable with respect to $\mathcal F_t$. It means that the sequential decision $D_t$ is made only based on the information that is currently available as recorded in $\mathcal F_t$.}
We remark that, under this framework, $S_{t+1}$ may be determined not only by the change detection results  suggested by statistical algorithms like the one proposed herein but also by the domain knowledge of the  test developers, which is common practice in the educational testing industry.
That is, the detection results only provide the testing program warnings on
potentially problematic items.
These items will be reviewed by the test developers and then decisions will be made on whether to delete some existing items from the item pool and whether to add new items.

\subsection{Proposed Compound Detection Procedure}\label{sec:dec}

We now propose a sequential change detection rule under the above general model.
Following the discussions above, at each time $t$, it seems natural to flag the items whose posterior probability $P(\tau_k < t \vert \mathcal F_t)$ is large, as a larger posterior probability implies a higher chance of having changed.
The question is, what cut-off value should we choose when making the decision? There is a trade-off behind this decision. On one hand, we would like to
detect as many post-change items as possible.
On the other hand, we want to avoid making many false detections of pre-change items, as false detections lead to unnecessary labor cost on item development, review  and maintenance.  In what follows, an optimization program will be proposed to balance this trade-off.


Recall that the goal of our monitoring procedure is to maintain the quality of an item pool that may be
measured by the proportion of post-change items in the pool at each time $t$. For example, the quality of $S_t$ can be measured by
$(\sum_{k\in S_t} 1_{\{\tau_k < t\}})/|S_t|$. The smaller the proportion, the better the quality of the item pool. As the change points $\tau_k$ for items $k \in S_t$ are unknown, this proportion is random. At a given time $t$, the best estimate of this quantity (under the mean squared error loss) is its conditional mean adaptive to the current information sigma-field $\mathcal F_t$.

The same quality measure can be used in the detection procedure. More precisely, let $D \subset S_t$ be the detection set after the $t$th test administration. The risk associated with the detection set $D$ adaptive to the information filtration can be measured by
\begin{equation}\label{eq:riskfunction}
R(D\vert \mathcal F_t) =  E\left(\frac{\sum_{k\in S_t\setminus D} 1_{\{\tau_k < t\}}}{\max\{|S_t \setminus D|,1\}}\big\vert \mathcal F_t \right),
\end{equation}
where the denominator is chosen so that $R(D\vert \mathcal F_t)$ is well-defined even when $D = S_t$.
The smaller value of $R(D\vert \mathcal F_t)$ implies the better quality of the undetected items
and thus a lower risk.
Therefore, a reasonable criterion is to control the risk $R(D\vert \mathcal F_t)$ to be below a pre-specified threshold (e.g., 1\%). Under this criterion, the expected proportion of post-change items in the undetected set is below the threshold and thus the item pool is
overall of high-quality. Consequently, in preparation for future test administrations, only the detected items need investigation.

\yc{Following the terminology of compound decision theory \citep{efron2012large}, we refer to the ratio $\big({\sum_{k\in S_t\setminus D} 1_{\{\tau_k < t\}}}\big)/{\max\{|S_t \setminus D|,1\}}$ as the False Nondiscovery Proportion (FNP) and the risk $R(D\vert \mathcal F_t)$ as the local False Nondiscovery Rate (FNR). Similarly, we define the False Discovery Proportion (FDP)
as ${\sum_{k\in D} 1_{\{\tau_k  \geq t\}}}/{\max\{1,|D|\}},$ that is the proportion of pre-change items in the detection set. The  local False Discovery Rate (FDR) is defined as the   conditional expectation of the FDP.
}

The proposed decision rule at time $t$ is to minimize the size of the detection set while controlling local FNR to be below a given threshold $\alpha$, i.e.,
\begin{equation}\label{eq:one-step}
D_t = \arg\min_{D\subset S_t} |D|, \mbox{~s.t.~} R(D\vert \mathcal F_t) \leq \alpha.
\end{equation}
Minimizing the detection set avoids making too many false detections of pre-change items, and the constraint on
local FNR ensures the detection of a sufficient number of post-change items.
\yc{The proposed decision rule $D_t$ can be obtained by is Algorithm~\ref{alg:one-step} below.  With given posterior probabilities, the computation of Algorithm~\ref{alg:one-step} is dominated by the sorting step whose complexity is $O(|S_t|\log|S_t|)$. The computation of posterior probabilities $W_{kt} = P(\tau_k < t\vert \mathcal F_t)$ under a specific change point model will be discussed in Section~\ref{subsec:changemodel} below.}
The detection rule \eqref{eq:one-step} is adaptive, in the sense that it makes use of up-to-date information $\mathcal F_t$. It is also compound, as the threshold on the posterior probabilities $W_{kt}$ is determined by the posterior probabilities of all the items in the current item pool $S_t$. \yc{The proposed procedure is optimal in the sense described in Proposition~\ref{prop:one-step}.}

\yc{
\begin{proposition}\label{prop:one-step}
The sequential decision rule $D_t$ given by Algorithm~\ref{alg:one-step} satisfies that $R(D_t|\mathcal{F}_t)\leq \alpha$.  In addition, for any other sequential decision rule $D'_t \subset S_t$ that is $\mathcal F_t$ measurable and satisfies $R(D_t'|\mathcal{F}_t)\leq \alpha$, we have $|D_t| \leq |D'_t|$ and
$$E\big(\sum_{k\in D_t} 1_{\{\tau_k \geq t\}}\big) \leq E\big(\sum_{k\in D'_t} 1_{\{\tau_k \geq t\}} \big).$$
\end{proposition}
}

\yc{Proportion~\ref{prop:one-step} implies that the proposed sequential decision rule minimizes the expected number of false detections of pre-change items at the current step, among all sequential decision rules that control the local FNR below the same level. Under some arguably more restrictive assumptions on the change point model, the proposed decision rule is not only optimal at the current step, but also uniformly optimal throughout the entire sequential decision process. This optimality result is given in Appendix B.}

{\centering
\renewcommand{\algorithmicrequire}{\textbf{Input:}}
\renewcommand{\algorithmicensure}{\textbf{Output:}}
\begin{algorithm}
\caption{Proposed detection rule.\label{alg:one-step}}
\begin{algorithmic}[1]
\Require  Threshold $\alpha$, the current item pool $S_t$, and posterior probabilities $(W_{kt})_{k\in S_t}$, where
$$W_{kt} = P(\tau_k < t\vert \mathcal F_t).$$
\State Sort the posterior probabilities in an ascending order. That is,
$$W_{k_1,t} \leq W_{k_2,t} \leq \cdots \leq W_{k_{|S_t|},t},$$
where $S_t = \{k_1, ..., k_{|S_t|}\}$. To avoid additional randomness, when there exists a tie ($W_{k_i,t} = W_{k_{i+1},t}$), we require $k_i <k_{i+1}$.

\State For $n = 1, ..., |S_t|$, define
$$V_{n} = \frac{\sum_{i=1}^n W_{k_i,t}}{n}.$$
and define $V_{0} =0$.
\State Find the largest $n \in \{0, 1, ..., |S_t|\}$ such that
$V_{n} \leq \alpha.$
\Ensure $D_t = S_t \setminus \{k_1, ..., k_n\}$ if $n\geq 1$ and $D_t = S_t$ if $n = 0$.
\end{algorithmic}
\end{algorithm}
}

\begin{remark}[Comparison with existing procedures]\label{rmk:compare}
 \yc{ We compare the proposed change detection framework with existing works on sequential item pool monitoring. Two different approaches are taken in the existing works. Specifically, \cite{veerkamp2000detection} and \cite{lee2016test} apply the CUSUM procedure to sequential data for each individual item and declare a change once the CUSUM statistic exceeds a pre-specified threshold. \cite{zhang2014sequential}, \cite{zhang2016monitoring} and \cite{choe2018sequential} test the pre-change hypothesis for each individual item at each time point. A change is declared if the p-value from the hypothesis test is smaller than a pre-specified threshold. The posterior probabilities in the current method
 play a similar role as the CUSUM statistics and  p-values in these works to measure the likelihood of each item having changed.

More specifically, we provide the connection between the posterior probabilities monitored by the proposed method and the p-values monitored by the sequential hypothesis testing method.
For a given item $k$ and at each time point $t$, the sequential hypothesis testing approach tests the null hypothesis $H_0: \tau_k \geq t$ versus the alternative hypothesis $H_1: \tau_k < t$. Following the routine of
frequentist hypothesis testing, a p-value is obtained based on some carefully designed test statistic. If the p-value is below a pre-specified threshold, $H_0$ is rejected and the item is declared to have changed. The proposed method tests the same hypotheses but takes a Bayesian approach. The prior distribution for $\tau_k$ implies the prior probabilities
$P(\tau_k \geq t)$ and $P(\tau_k < t)$ for the null and alternative hypotheses, respectively. Given these
prior probabilities, together with the  pre- and post-change distributions of data, the Bayes formula provides us the posterior probabilities of the null and alternative hypotheses. Bayesian hypothesis testing makes decision based on these posterior probabilities. We refer the readers to \cite{wagenmakers2010bayesian} for a discussion on Bayesian hypothesis testing and a comparison with the frequentist approach.
In this sense, the proposed method can be viewed as a Bayesian version of the sequential hypothesis testing method.




 Although both the CUSUM and the sequential hypothesis testing methods can effectively detect post-change items, they do not provide an estimate of the number/proportion of post-change items among the undetected ones. Consequently, they cannot directly assess and control the quality of the item pool. These methods may still be able to control a similar risk as in the proposed method by tuning the corresponding threshold for declaring changes, but choosing a suitable threshold for this purpose is a challenging task that may require external information about the number/proportion of post-change items in the pool.

 In contrast, the proposed method can directly control the quality of item pool by controlling the local FNR. As a price, it requires to know the Bayesian model for change points,
 including
 the prior distribution for the change points, and the pre- and post-change distributions for the data streams. As will be explained in Section~\ref{sec:unknown}, the proposed method can be extended to
 control the local FNR given only some partial information about the Bayesian change point model.

  }
\end{remark}

\begin{remark}[Application to continuous testing]

{Item leakage may be more likely to occur in continuous testing with
Computerized Adaptive Testing (CAT), MST and fixed-form testing designs.} Continuous testing implies more frequent item usage from an item pool, for which monitoring item pool quality may be even more important. In particular, most of the existing works on the sequential detection of item changes are developed under a continuous CAT setting \citep{veerkamp2000detection,zhang2014sequential, zhang2016monitoring, choe2018sequential} or under a continuous MST or linear testing setting
\citep{lee2016test}.


We point out that the proposed framework is very general that can also be applied to monitoring the item pool of any continuous test. {A setting for frequent but non-continuous fixed-form tests
is adopted in Section~\ref{subsec:general}  for the ease of exposition, as our main focus is to introduce the  new compound sequential decision framework.} To apply the proposed framework to continuous testing, the meaning of each time point $t$ needs to be slightly different from that in the existing works  concerning a CAT setting. That is, most of the existing works   concerning a CAT setting  focus on individual items and examinees, where each time point for an item corresponds to its admission to an examinee. For example, for a given item, $t = 20$ means the item having been administered to 20 examinees. Consequently, the same $t$ does not necessarily correspond to the same calendar time for different items, as different items may have different exposure rates. This choice of time $t$ is thus not suitable for defining our compound risk that measures the item pool quality at some point of calendar time.
To apply the proposed method to continuous testing,
we can let each time point correspond to a fixed period of calendar time, for
example, one day or half a day. The duration of the time period may be chosen based on the
test volume to allow adequate sample sizes in computing the monitoring statistics.
The monitoring statistics at a time point are constructed based on all the item responses collected during the corresponding period.
The posterior distributions of change points can be updated based on the monitoring statistics. The compound risk at each time point can thus be defined and compound sequential decisions can be made accordingly. {See Appendix~E for a further discussion  on applying the proposed method to continuous testing.}





\end{remark}


\subsection{A Specific Change Point Model}\label{subsec:changemodel}

We now provide a specific model for illustration. We assume that the change point $\tau_k$ satisfies
$$|\{s: k \in S_s^*, s\leq \tau_k \}| =  \gamma_k,$$
where $\gamma_1$, $\gamma_2$, ... are independent, each of which follows a geometric distribution.
That is,

\begin{equation}\label{eq:geo-prior}
	\yc{P(\gamma_k = m) = (1-\rho_k)^{m-1}\rho_k,~ m=1, 2, ...,}
\end{equation}
where $\rho_k$ is an item-specific parameter in the open interval $(0,1)$.
This model implies that, on average, the status of an item changes (e.g., being leaked) after being used in $1/\rho_k$ test administrations (i.e., exposures).
We point out that the geometric distribution is widely used in the Bayesian formulation of sequential change detection because of its memoryless property. The proposed methods can be extended to other prior distributions for the change points.

We may further assume that both the pre- and post-change distributions are univariate normal.
Specifically, as will be justified by the possible choices of the monitoring statistic as in Section~\ref{sec:stat},
we assume that the pre-change distribution $p_{kt}$ is standard normal. The post-change distribution $q_{kt}$ is $N(\mu_{kt}, 1)$. For now, we consider the case where $\rho_k$ and $\mu_{kt}$ are both known and leave the unknown case to Section~\ref{sec:unknown}.

Under this model, the posterior distribution $P(\tau_k < t \vert \mathcal F_t)$ can be computed in an analytic form. To simplify the notation, we denote $W_{kt} = P(\tau_k < t \vert \mathcal F_t)$.
 This posterior probability can be obtained by a simple updating rule, summarized in the following proposition.
\begin{proposition}\label{prop:update-post}
	Assume $\gamma_k$ follows a geometric prior described in \eqref{eq:geo-prior} and let
$e_{kt} = \sum_{i=1}^t 1_{\{k\in S_i^*\}}$ be the number of exposure of  item $k$ up to time $t$.
Then $W_{kt}=\frac{U_{kt}}{U_{kt}+1/\rho_k}$, where $U_{kt}$
is computed through the following updating rule,
		\begin{equation}\label{eq:update}
 \text{ if } e_{kt} \leq 1,		U_{kt}=0,
\mbox{~else~}		U_{kt}=		
		\begin{cases}
			U_{k,t-1} \text{ if } k \in S_t\setminus S_t^*,\\
			(1+U_{k,t-1})\frac{q_{kt}(X_{kt})}{(1-\rho_k)p_{kt}(X_{kt})}
			 \text{ if }k \in  S_t^*.\\
		\end{cases}
\end{equation}
Recall that $p_{kt}$ and $q_{kt}$ are
the density functions of the pre- and post-change distributions, respectively, defined in \eqref{eq:dists}.
In particular, if $p_{kt}\sim N(0,1)$ and $q_{kt}\sim N(\mu_{kt},1)$, then
	\begin{equation}
			 \text{ if } e_{kt} \leq 1,		U_{kt}=0,
\mbox{~else~}	
		U_{kt}=		
		\begin{cases}
			U_{k,t-1} \text{ if } k \in S_t\setminus S_t^*\\
			\frac{1}{1-\rho_k}(1+U_{k,t-1})\exp\big\{\mu_{kt} X_{kt}-\frac{\mu_{kt}^2}{2}\big\} \text{ if }k \in  S_t^*.
		\end{cases}
	\end{equation}

\end{proposition}
In the above proposition, the statistic $U_{kt}$ is a modification of a classic sequential change detection statistic \citep{shiryaev1963optimum} which gives
optimal sequential change detection for a single data stream under a Bayesian decision framework. \yc{We point out that  $U_{kt}$ is not updated until item $k$ is exposed at least twice (i.e., $e_{kt} > 1$), because we do not allow $\tau_k = 0$. According to \eqref{eq:update},  the update of the posterior probabilities $W_{kt}$ is straightforward when the pre- and post-change distributions are known. These distributions are not necessarily normal, though it may be convenient to make the normality assumption in the current application
as discussed in Section~\ref{sec:stat}.}





\section{When Model is not Completely Known}\label{sec:unknown}

Now we consider the situation
in which only partial information is available about the change point model. More precisely, we focus on the specific change point model given in Section~\ref{subsec:changemodel}.
\yc{Recall that the change point $\tau_k$ follows
a geometric distribution \eqref{eq:geo-prior} with parameter $\rho_k$.}
It is further assumed that the pre-change distribution $p_{kt}$ is known, for example, a standard normal distribution.
In addition, it is assumed that the post-change distribution $q_{kt}$ can be parameterized as
$$q_{kt}(x) = h_{kt}(x\vert \boldsymbol\pi_k),$$
where $h_{kt}$ is a known function and
$\boldsymbol\pi_k$ is an item-specific parameter vector.
This parametrization of $q_{kt}$ will be justified under an item response theory model in Section~\ref{sec:stat}.



In practice,  $\rho_k$ and $\boldsymbol\pi_k$ are unknown, but prior information may be available. Specifically, $1/\rho_k$ represents the average number of exposures (i.e., number of times the item is used) for the item to change.  Although it is hard to know $1/\rho_k$ precisely, a reasonable lower bound is often available, which leads to an upper bound for $\rho_k$, denoted by $\overline{\rho} \geq \rho_k$, for all $k$.
In addition, as will be further justified in Section~\ref{sec:stat}, we assume that $\boldsymbol \pi_k \in \Theta$, where $\Theta$
is a known compact set.
%
%
%
%
In what follows, we propose a method that controls the compound risk $R(D_t \vert \mathcal F_t)$, when only knowing $\bar \rho$,
$p_{kt}$, and $\Theta$.



Let $W_{kt}(\rho,\boldsymbol\pi)$ denote the posterior probability $P(\tau_k<t|\mathcal F_t)$ when the underlying parameters are $\rho_k = \rho$ and $\boldsymbol\pi_k = \boldsymbol\pi$, and define $\overline{W}_{kt}$ as
\begin{equation}\label{eq:max-w}
	\overline{W}_{kt} = \sup_{\rho \in (0,\overline{\rho}],\boldsymbol\pi \in \Theta} W_{kt}(\rho,\boldsymbol\pi).
\end{equation}
\yc{Note that $\overline{W}_{kt}$ is not a posterior probability, but an upper bound of the posterior probability $W_{kt}(\rho_k, \boldsymbol\pi_k)$.}
We now provide Algorithm~\ref{alg:one-step-unknown} that
replaces ${W}_{kt}$
in Algorithm~\ref{alg:one-step}
by $\overline{W}_{kt}$. As shown in Theorem~\ref{thm:unknown}, the decision rule given by Algorithm~\ref{alg:one-step-unknown}
controls the compound risk $R(D_t \vert \mathcal F_t)$ at any time $t$.


{\centering
\renewcommand{\algorithmicrequire}{\textbf{Input:}}
\renewcommand{\algorithmicensure}{\textbf{Output:}}
\begin{algorithm}
\caption{Proposed detection rule.\label{alg:one-step-unknown}}
\begin{algorithmic}[1]
\Require  Threshold $\alpha$, the current item pool $S_t$, and $\overline{W}_{kt}$ defined in \eqref{eq:max-w}.
\State Sort the $\overline{W}_{kt}$ in an ascending order. That is,
$$\overline{W}_{k_1,t} \leq \overline{W}_{k_2,t} \leq \cdots \leq \overline{W}_{k_{|S_t|},t},$$
where $S_t = \{k_1, ..., k_{|S_t|}\}$. To avoid additional randomness, when there exists a tie ($\overline{W}_{k_i,t} = \overline{W}_{k_{i+1},t}$), we require $k_i <k_{i+1}$.

\State For $n = 1, ..., |S_t|$, define
$$V_{n} = \frac{\sum_{i=1}^n \overline{W}_{k_i,t}}{n}.$$
and define $V_{0} =0$.
\State Find the largest $n \in \{0, 1, ..., |S_t|\}$ such that
$V_{n} \leq \alpha.$
\Ensure $D_t = S_t \setminus \{k_1, ..., k_n\}$ if $n\geq 1$ and $D_t = S_t$ if $n = 0$.
\end{algorithmic}
\end{algorithm}
}

\begin{theorem}\label{thm:unknown}
Suppose that $\rho_k \leq \overline{\rho}$ and $\boldsymbol\pi_k \in \Theta$ for all $k$.
Then the proposed decision rule given in Algorithm~\ref{alg:one-step-unknown}  guarantees that $R(D_t |\mathcal{F}_t)\leq \alpha$ for all $t=1,2,\dots$
\end{theorem}

It remains to find a way to compute $\overline{W}_{kt}$, as it is defined by an optimization over an iteratively defined object.
Proposition~\ref{prop:w-extended} below provides guidance to this problem.

\begin{proposition}\label{prop:w-extended}
Let $e_{kt} = \sum_{i=1}^t 1_{\{k\in S_i^*\}}$ be the number of exposures of item $k$ up to time $t$.
	Define $U_{kt}(\overline{\rho},\boldsymbol\pi)$ according to the following iterations,
	\begin{equation}\label{eq:iteration}
 \text{ if } e_{kt} \leq 1,			U_{kt}(\overline{\rho},\boldsymbol\pi)=0,
\mbox{~else~}
		U_{kt}(\overline{\rho},\boldsymbol\pi)=		
		\begin{cases}
			U_{k,t-1}(\overline{\rho},\boldsymbol\pi) \text{ if } k \in S_t\setminus S_t^*,\\
			\frac{1}{1-\bar\rho}(1+U_{k,t-1}(\overline{\rho},\boldsymbol\pi))\frac{h_{kt}(X_{kt}\vert \boldsymbol\pi)}{p_{kt}(X_{kt})} \text{ if }k \in  S_t^*.
		\end{cases}
\end{equation}
Let $\overline{R}_{kt}= \sup_{\boldsymbol\pi\in \Theta} U_{kt}(\overline{\rho},\boldsymbol\pi)$. Then, $\overline{W}_{kt}=\frac{\overline{R}_{kt}}{\overline{R}_{kt}+1/\overline{\rho}}$.
\end{proposition}
\yc{It can be shown that $W_{kt}(\rho,\boldsymbol\pi)$ is monotone increasing with respect to $\rho$. Therefore, to obtain $\overline{W}_{kt}$, Proposition~\ref{prop:w-extended} plugs $\overline{\rho}$ into $W_{kt}(\rho,\boldsymbol\pi)$.}
When the dimension of $\boldsymbol\pi$ is very low (e.g., one or two),
we can discretize the set $\Theta$ by grid points, update $U_{kt}(\overline{\rho},\boldsymbol\pi)$ on the grid points in parallel, and then approximate $\overline{W}_{kt}$ accordingly. When the number of parameters in $\boldsymbol\pi$ is not very low, by making use of \eqref{eq:iteration}, the gradient of $U_{kt}(\overline{\rho},\boldsymbol\pi)$ with respect to $\boldsymbol\pi$ can be computed iteratively. Thus,  $\overline{R}_{kt}$ can be computed, for example, by a gradient ascent algorithm.

\section{Monitoring Statistic in IRT-based Testing}\label{sec:stat}
In principle, the monitoring statistic $X_{kt}$ can be any statistic whose distribution is different before and after the change point. In practice, we suggest to choose $X_{kt}$ to be a Wald-type statistic, so that we can approximate
the pre- and post-change distributions by normal distributions to simplify the computation.  In what follows, we give an example of such a monitoring statistic and explain how confounding due to changing examinee population is adjusted in this statistic.


\subsection{Standardized Item Residual Statistic}

Let $N_t$ be the number of people taking the test at time $t$. Let $Y_{ktn} \in \{0,1\}$ denote the $n$th examinee's response to item $k$ at time $t$, where $Y_{ktn} = 1$ indicates a correct response and $Y_{ktn} = 0$ otherwise.

Let $\bar{Y}_{kt}=(\sum_{n=1}^{N_t} Y_{ktn})/N_t$ be the percent correct for item $k$ at time $t$.
Suppose that a change has not yet occurred. Then $\bar{Y}_{kt}$ has expected value
$\xi^0_{kt} := E\big(\bar{Y}_{kt}|\tau_k\geq t\big)$. If $\xi^0_{kt}$ is known and let $SE(\bar{Y}_{kt})$ be the standard error of $\bar{Y}_{kt}$,
then the standardized residual $(\bar{Y}_{kt} - \xi^0_{kt})/SE(\bar{Y}_{kt})$ is approximately standard normal when the change has not yet occurred and $N_t$ is sufficiently large. Note that this statistic can adjust for the examinee population change when $\xi^0_{kt}$ is defined under the IRT framework.

In practice, we typically do not know $\xi^0_{kt}$. Specially, when the examinee population changes overtime, $\xi^0_{kt}$ needs to be estimated based on both historical information and data from the $t$th test administration. Now suppose that we have a consistent estimator
of $\xi^0_{kt}$, denoted as  $\hat \xi^0_{kt}$. The way obtaining $\hat \xi^0_{kt}$ will be discussed in Section~\ref{subsec:IRT}.
Then the Standardized Item Residual (SIR) statistic is defined as
$$X_{kt} = \frac{\bar{Y}_{kt}- \hat \xi^0_{kt}}{SE(\bar{Y}_{kt}- \hat \xi^0_{kt})},$$
where $SE(\bar{Y}_{kt}- \hat \xi^0_{kt})$ is the standard error of the numerator.
Under mild conditions and given that $\tau_k \geq t$,
$X_{kt}$ is approximately
standard normal for sufficiently large $N_t$. Similarly,  given that $\tau_k < t$,
$X_{kt}$ is asymptotically $N(\mu_{kt}, 1)$, where $\mu_{kt}$ characterizes the mean change of the SIR statistic.
\yc{The normal approximation can be justified in sense that $X_{kt}$ can be decomposed as $X_{kt} =\mu_{kt} + Z + o_p(1)$, which $Z$ is a standard normal random variable. This holds even when $\mu_{kt}$ diverges.
When focusing on change points due to item leakage, it is reasonable to further assume that  $\mu_{kt} >0$.}

In general, the SIR statistics $X_{kt}$, $k \in S_t$, are not independent as assumed in our Bayesian change point model.
Specifically, the SIR statistics constructed under an IRT model tend to have weak positive correlations, brought by individual-specific latent factors.
As shown in Section~\ref{sec:sim},
with such dependence, the proposed method still controls the compound risk.


\subsection{SIR Statistic under IRT Framework}\label{subsec:IRT}

In what follows, we describe an IRT model for item response data $Y_{ktn}$. Under this model, the SIR statistic
$X_{kt}$ can be computed and the mean change $\mu_{kt}$ can be expressed as a function of the parameters in the IRT model.

%


\paragraph{IRT model for pre-change items.} IRT provides a popular method in educational testing for linking different test administrations with potentially different examinee populations. In the current context, it is sensible to model item responses using a unidimensional IRT model for pre-change items.
A unidimentional IRT model assumes that each item  is characterized by one or multiple parameters that do not change over time, denoted by $\bbb_k$, and that each examinee
$n$ in test administration $t$
is characterized by one parameter, denoted by $\theta_{tn}$. The parameter $\theta_{tn}$ is typically interpreted as the ability of the examinee.

Under this model, the probability of an examinee answering item $k$ correctly is completely determined by the item parameters and the person parameter in the form
$$P(Y_{ktn} = 1 \vert \theta_{tn}, \bbb_k) = f(\theta_{tn} \vert \bbb_k),$$
where $f$ is a pre-specified inverse-link function that is monotonically increasing in $\theta_{tn}$. For example, one of the most commonly used models in educational testing is the so-called two-parameter logistic (2PL) model \citep{birnbaum1968some}. Under the 2PL model,
$$f(\theta_{tn} \vert \bbb_k) = \frac{\exp(\beta_{k0} + \beta_{k1} \theta_{tn})}{1+\exp(\beta_{k0} + \beta_{k1} \theta_{tn})},$$
where  $\beta_{k0}$ and $\beta_{k1} > 0$ are known as the easiness and discrimination parameters, respectively, and $\bbb_k = (\beta_{k0}, \beta_{k1})$. Given the person and item parameters, an examinee's responses to different items are assumed to be independent, known as the local independence assumption.

In this paper, the item parameters $\bbb_k$ are treated as fixed parameters that do not change over time. The person parameters $\theta_{tn}$ are treated as random variables.
Specifically, $\theta_{t1}, ..., \theta_{t,N_{t}}$ are assumed to be i.i.d. samples from a distribution $N(m_t, 1)$, where the mean is time dependent to reflect the population change over time (e.g., seasonal effect)  and the variance is fixed to be one to bypass the scale-indeterminacy.
Under the IRT model, the expected percent correct can be calculated as
\begin{equation}\label{eq:xi}
\xi^0_{kt} = \int f(\theta \vert \bbb_k) \frac{1}{\sqrt{2\pi}} \exp(-(\theta - m_t)^2/2) d\theta.
\end{equation}

\paragraph{IRT model for post-change items.}  We now describe an IRT model for response data involving item preknowledge.
The same model has been adopted in \cite{lee2016test} who developed
CUSUM statistics based on this model to monitor item performance and detect item
preknowledge.
For each $k$ satisfying $t > \tau_k$,
we use $\eta_{ktn} \in \{0,1\}$ to denote preknowlege about the item. That is, $\eta_{ktn}  = 1$ indicates that the examinee has preknowledge about the item, and $\eta_{ktn}  = 0$ otherwise. We assume that
the indicators $\eta_{ktn}$ are i.i.d., following a Bernoulli distribution with $P(\eta_{ktn} = 1) = \pi_k$, and $\eta_{ktn}$
is independent of $\theta_{tn}$.
That is, whether an individual has preknowlege about a leaked item is independent of his/her ability.
We further assume that when examinee $n$ has preknowledge about item $k$, his/her answer is always correct. That is,
$Y_{ktn} = 1$, given that $\eta_{ktn} = 1$. Finally, it is assumed that when  examinee $n$ does not cheat on item $k$,
i.e., $\eta_{ktn} = 0$, then $Y_{ktn}$ given $\eta_{ktn}$ and $\theta_{tn}$ still follows the same IRT model as if the item has not changed. This post-change model yields
$\xi^1_{kt} = (1-\pi_k)\xi^0_{kt} + \pi_k$ and therefore $\xi^1_{kt} - \xi^0_{kt} = \pi_k(1-\xi^0_{kt})$.

\paragraph{Estimation of $\xi^0_{kt}$.} 
In practice, the expected pre-change percent correct $\xi^0_{kt}$ is unknown due to the unknown
mean $m_t$ which needs to be estimated based on data from the pre-change items in the
$t$th test administration.
Suppose that there exists a non-empty subset $S^\dagger_t \subset S_t^*$ that is known to only contain pre-change items. For example, $S^\dagger_t$ can be the items that have not been exposed before. Then $m_t$ can be consistently estimated by maximizing the marginal likelihood
\begin{equation}\label{eq:mle}
\hat m_t = \argmax_{m_t} \sum_{n=1}^{N_t} \log\left(\int  \prod_{k\in S_t^\dagger} f(\theta \vert \bbb_k)^{Y_{ktn}} (1-f(\theta \vert \bbb_k))^{1-Y_{ktn}} \frac{1}{\sqrt{2\pi}} \exp(-(\theta - m_t)^2/2) d\theta\right).
\end{equation}
Accordingly, $\xi^0_{kt}$ can be estimated by plugging in $\hat m_t$. The standard error $SE(\bar{Y}_{kt}- \hat \xi^0_{kt})$ can also be computed easily; see the details in Appendix~B.

%

\paragraph{Pre- and post-change distributions of SIR statistic.}
The SIR statistic $X_{kt}$ constructed above is approximately standard normal when $\tau_k \geq t$, and is approximately normal
$N(\mu_{kt},1)$ when $\tau_k < t$, where
$$\mu_{kt} = \frac{\xi^1_{kt} - \xi^0_{kt}}{\sqrt{Var(\bar{Y}_{kt}- \hat \xi^0_{kt})}}.$$
The value of $\xi^1_{kt}$ is determined by the post-change model introduced above and the value of  $\xi^0_{kt}$ is determined by the pre-change model that can be estimated from data.
Given the leakage proportion $\pi_k$, $\mu_{kt}$ can be approximated by
$$\hat \mu_{kt}(\pi_k) = \frac{\pi_k(1-\hat\xi^0_{kt})}{{SE(\bar{Y}_{kt}- \hat \xi^0_{kt})}}.$$

In practice,  $\pi_k$ is usually unknown and hard to estimate, though it may be known by prior knowledge that
$\pi_k$s locate in a certain interval $\Theta$. With such information, we can run Algorithm~\ref{alg:one-step-unknown} with
\begin{equation}\label{eq:postest}
h_{kt}(x\vert \pi_k) = \frac{1}{\sqrt{2\pi}} \exp\left(-\frac{(x - \hat \mu_{kt}(\pi_k))^2}{2}\right).
\end{equation}

\section{Simulation Study}\label{sec:sim}

\subsection{Study I}

\paragraph{Known change point model.}
We start with a simple simulation setting to illustrate the proposed method.
We consider an item pool originally containing $|S_1| = 500$ items. During the process, once a subset of items are detected, they will be removed, and the same number of new items will be added to ensure $|S_t| = 500$ for all $t$. We also assume that 50 items are randomly selected from the item pool for each test administration, i.e., $|S_t^*| = 50$.

The parameter $\rho_k$ in the change time distribution is generated from a uniform distribution over the interval $[0, 0.1]$ for different $k$. It is further assumed that the monitoring statistic $X_{kt}$ follows $N(0, 1)$ when $\tau_k \geq t$, and follows $N(\mu_k, 1)$ when $\tau_k < t$. We generate $\mu_k$ from a uniform distribution over the interval $[1,2]$ for different $k$.

We investigate the situation when $\rho_k$ and $\mu_k$ are known. We run 1000 independent simulations. In each simulation, we apply Algorithm~\ref{alg:one-step}, for $t = 1, ..., 50$, where the threshold $\alpha$ for the compound risk is set to be 0.01. To evaluate the method, three metrics are calculated at each time $t$, including (a) the  FNP  $$\frac{\sum_{k\in S_t\setminus D_t} 1_{\{\tau_k  < t\}}}{\max\{1,|S_t \setminus D_t|\}},$$ (b) the  FDP
$$\frac{\sum_{k\in D_t} 1_{\{\tau_k  \geq t\}}}{\max\{1,|D_t|\}},$$
and (c) the number of detections $|D_t|$. Our results are shown in Figure~\ref{fig:sim1_1}, where panels (a)--(c) show the three metrics, respectively. In each panel, the $x$-axis shows time $t$ and the $y$-axis shows the
5\%, 25\%, 50\%, 75\%, and 95\% quantiles of the empirical distribution of the metric based on 1000 independent simulations.

We take a closer look at the medians of these metrics over time.
The median FNP is zero at the beginning, because
at time $t = 1$ all the items are new (i.e., never exposed before).
It increases as time goes on and stabilizes at the targeted level  0.01 after about 10 time points.
The median FDP is also zero at the beginning but then increases
dramatically. It stabilizes around the level 0.8 after about 20 time points. Finally, the median detection size
also increases with time $t$ and becomes stable around 10 after about 20 time points. 


\begin{figure}
  \centering
  \includegraphics[scale = 0.4]{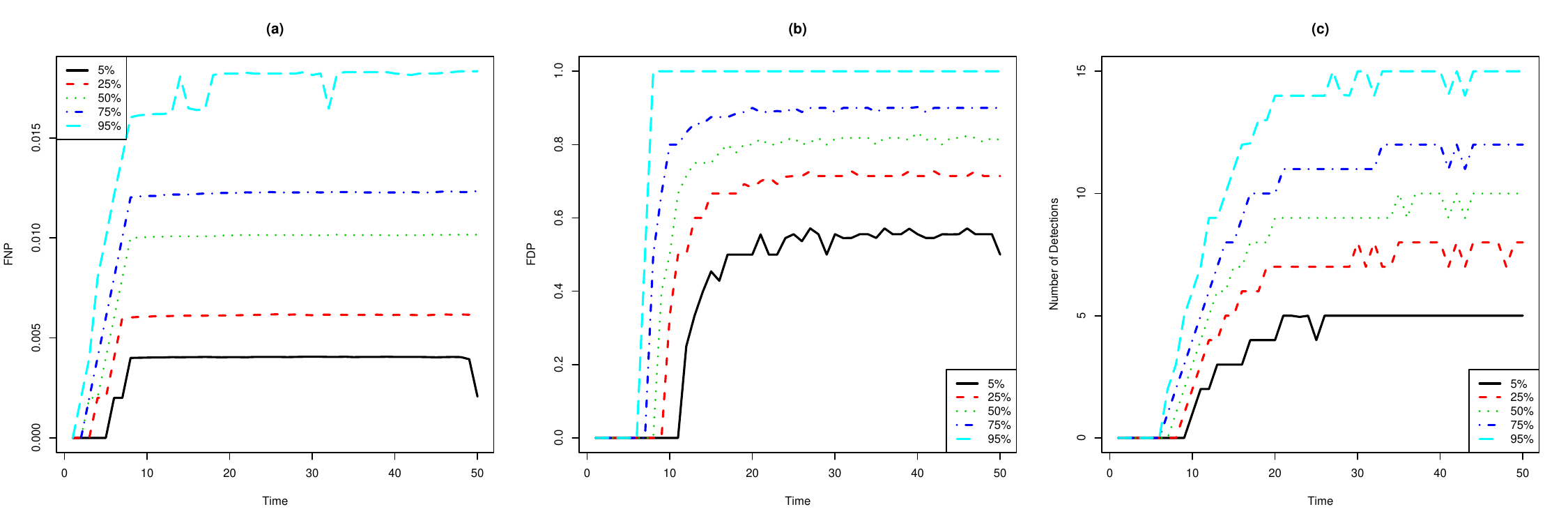}
  \caption{The performance of Algorithm~\ref{alg:one-step} under model correct specification. Panels (a)--(c) show the FNP, FDP, and number of detections, respectively. In each panel, the $x$-axis shows time $t$ and the $y$-axis shows the 5\%, 25\%, 50\%, 75\%, and 95\% quantiles of the empirical distribution of the metric based on 1000 independent simulations. }\label{fig:sim1_1}
\end{figure}

\paragraph{Unknown change point model.} We now look at the situation when
$\rho_k$ and $\mu_k$ are unknown under the same simulation setting as above.
We again run 1000 independent simulations. In each simulation, we apply Algorithm~\ref{alg:one-step-unknown}, for $t = 1, ..., 50$, where the threshold $\alpha=0.01$.
When applying Algorithm~\ref{alg:one-step-unknown}, we only know that the pre-change distribution is $N(0,1)$,
$\rho_k \in (0, 0.1]$ and $\mu_k \in [1,2]$. The results are shown in Figure~\ref{fig:sim1_2} which take a similar form as those given
in Figure~\ref{fig:sim1_1}. As we can see, when the change point model is unknown, the decision given by Algorithm~\ref{alg:one-step-unknown} still controls the FNP under the targeted level. However, when comparing the current results with those from Algorithm~\ref{alg:one-step} above, we see that the decision rule given by Algorithm~\ref{alg:one-step-unknown} is more conservative, which is a price paid for not knowing the parameters in both the geometric distribution for change points and the post-change distribution.  As a result of the conservativeness in the decision (i.e., smaller FNP), the FDP and the number of detections are both larger comparing with the results in Figure~\ref{fig:sim1_1}. 

\begin{figure}
  \centering
  \includegraphics[scale = 0.4]{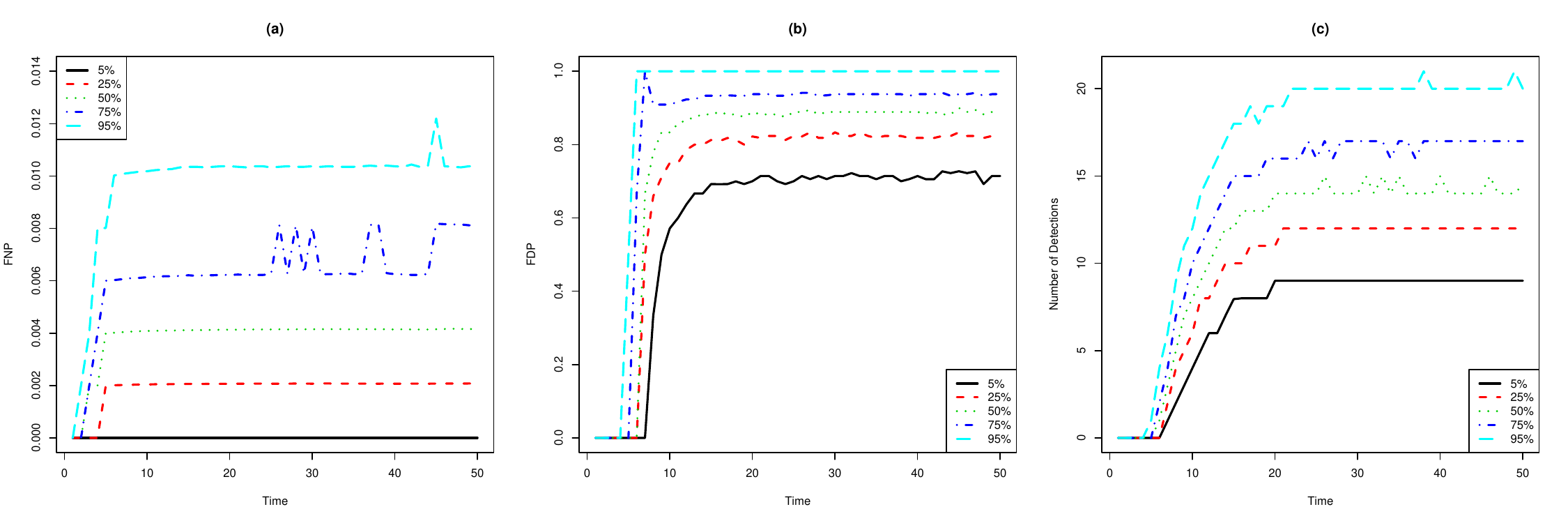}
  \caption{The performance of Algorithm~\ref{alg:one-step-unknown} under model correct specification.  The plots can be interpreted the same as those in Figure~\ref{fig:sim1_1}.}\label{fig:sim1_2}
\end{figure}

\paragraph{Under model misspecification.} We further look at the situation when the change point model is misspecified. Specifically, we consider the case when the monitoring statistics $X_{kt}$, $k\in S^*_t$, are correlated. More precisely, we assume $(X_{kt}: k \in S^*_t)$ given $\tau_k$, $k\in S^*_t$ is multivariate normal distribution, for which the marginal pre-change distribution of $X_{kt}$ is still standard normal and the marginal post-change distribution is $N(\mu_k,1)$, and the covariance between $X_{kt}$ and $X_{k't}$ is 0.1 for $k\neq k'$. The rest of the setting is the same as above.

We apply Algorithm~\ref{alg:one-step}, assuming that $\rho_k$ and $\mu_k$ are known and pretending that the data streams are independent. The results are shown in Figure~\ref{fig:sim1_3}. Comparing the results in Figures~\ref{fig:sim1_1} and \ref{fig:sim1_3}, it seems that the
performance metrics are only slightly affected when we simply run Algorithm~\ref{alg:one-step},
ignoring the weak positive dependence between the data streams.
We further apply Algorithm~\ref{alg:one-step-unknown}, with knowledge that all the $\rho_k$s lie in the interval $(0, 0.1]$ and that all the $\mu_k$s lie in the interval $[1,2]$. Again, when running Algorithm~\ref{alg:one-step-unknown}, we pretend that the data streams are independent.
The results are shown in Figure~\ref{fig:sim1_4}. Similar to the results when given the known model, the effect of ignoring the weak positive dependence in data also seems small when the change point model is not completely known. 

\begin{figure}
  \centering
  \includegraphics[scale = 0.4]{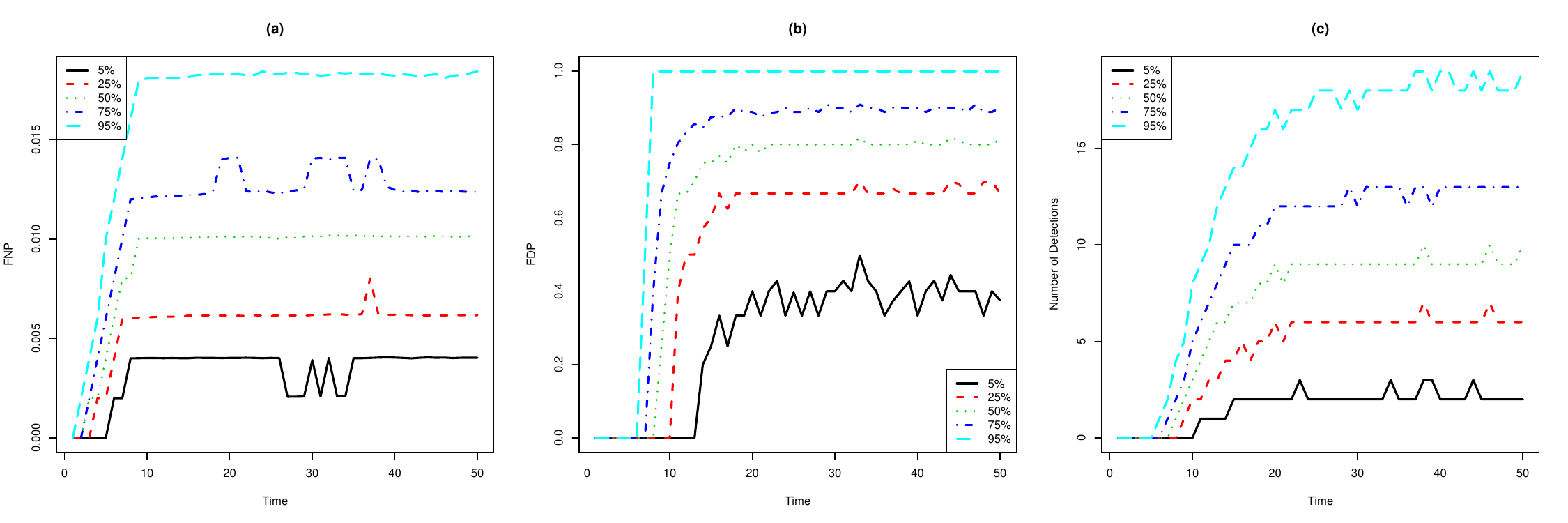}
  \caption{The performance of Algorithm~\ref{alg:one-step} under model mispecification.  The plots can be interpreted the same as those in Figure~\ref{fig:sim1_1}.}\label{fig:sim1_3}
\end{figure}

\begin{figure}
  \centering
  \includegraphics[scale = 0.4]{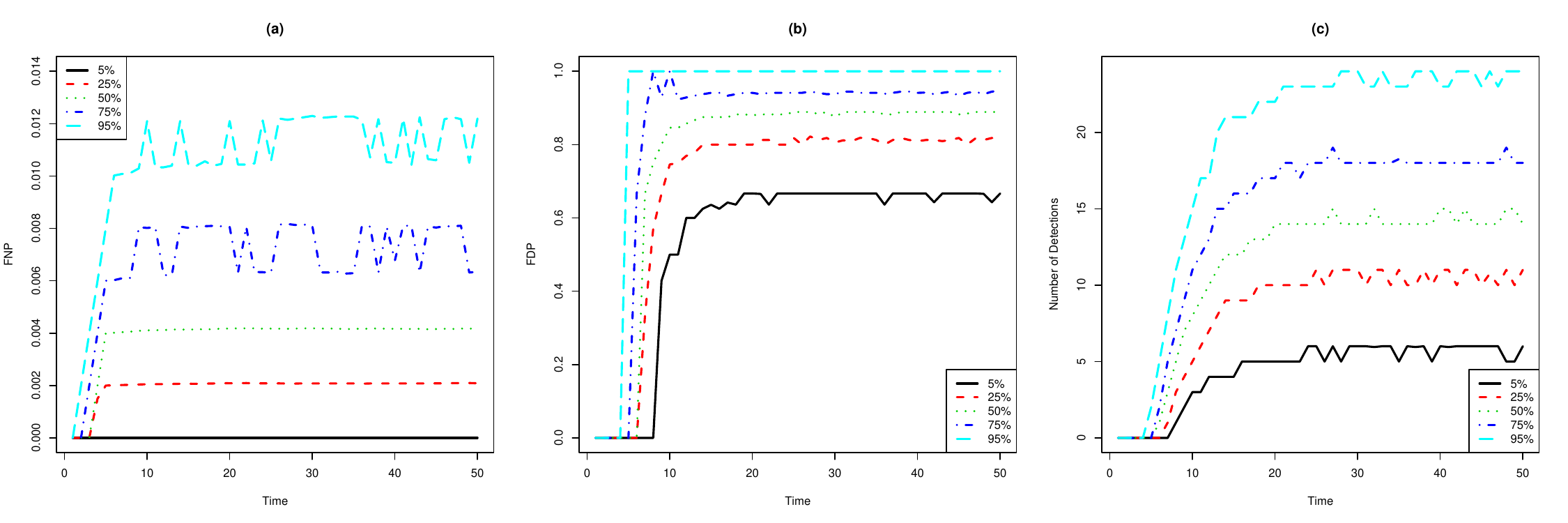}
  \caption{The performance of Algorithm~\ref{alg:one-step-unknown} under model mispecification.  The plots can be interpreted the same as those in Figure~\ref{fig:sim1_1}.}\label{fig:sim1_4}
\end{figure}


\subsection{Study II: Educational Testing with Time-varying Population}\label{subsec:study2}
\paragraph{Simulation setting.} We now evaluate the proposed method under an IRT setting that  mimics operational tests. The setting is almost the same as above, except for the way the monitoring statistics are obtained. More precisely, pre-change item response data are simulated using the 2PL model introduced in Section~\ref{subsec:IRT}. Each item $k$ is assumed to be associated with item parameters $\beta_{k0}$ and $\beta_{k1}$, where the discrimination parameter $\beta_{k1}$ is generated from a uniform distribution $[1, 1.5]$, and the easiness parameter $\beta_{k0}$ is generated from a uniform distribution $[-2, 2]$. We assume that the item parameters are known in our sequential decision procedure, because in practice these parameters can usually be accurately pre-calibrated, before their operational use.
At each time $t$, the number of examinees $N_t$ is generated from a uniform distribution over the set $\{1001, 1002, ..., 3000\}$.
Each examinee $n$ at time $t$ is associated with an ability parameter $\theta_{tn}$, generated from  normal distribution $N(m_t, 1)$, where the population mean $m_t$ is generated from a uniform distribution over the interval $[-0.5, 0.5]$.
The post-change item response data are simulated using the mixture model described in Section~\ref{subsec:IRT}.
The leakage proportion $\pi_{k}$ is generated from a uniform distribution over the interval $[0.05, 0.1]$.

The simulated item pool originally contains $|S_1| = 500$ items. During the process, once a subset of items are detected, they will be removed, and the same number of new items will be added to maintain the size of the item pool. In addition, we add new items to $S_t$ when necessary to ensure that
$S_t$ always contains at least 5 new items that have not been exposed before.
This is because, we include at least 5 new items in
$S_t^*$
for the estimation of $m_t$, the mean of the population ability distribution at time $t$.
The rest of the simulation setting is the same as that of Study I.

We construct an SIR statistic for each data stream using the method introduced in Section~\ref{subsec:IRT}. Both the pre- and post-change distributions are approximated by normal distributions. \yc{We point out that the signal under the current setting is stronger than that under Study 1. In particular, the 25\%, 50\% and 75\% quantiles
of the empirical distribution for $\hat \mu_{kt}(\pi_k)$  are  2.3,  3.9, and  5.1, respectively.}
Two cases are considered. In the first case (Case I), both parameters $\rho_k$ and leakage proportions $\pi_k$ are treated as known. We run Algorithm~\ref{alg:one-step} with the pre-change distribution being standard normal and the post-change distribution being estimated as in \eqref{eq:postest}.
In the second case (Case II), both $\rho_k$ and  $\pi_k$ are treated as unknown, which is usually the case in practice. We assume that we only know these parameters lying in the intervals
$(0, 0.1]$ and $[0.05, 0.1]$, respectively.
We run Algorithm~\ref{alg:one-step-unknown}, with the pre-change distribution being standard normal and the post-change distribution being estimated as in \eqref{eq:postest}.


\paragraph{Results.} 

The results for Case I are presented in Figures~\ref{fig:sim2_1}. As we can see, the proposed method still performs reasonably well under this setting. 
Specifically, the median FNP is slightly larger than the targeted level 0.01 after 20 time points, but it never exceeds 0.013. This slight overshoot is likely due to the normal approximation and the estimation of the population distribution at each time point.
The FDP is quite small, with the median FDP always being zero. This is because, the signal of the change points is quite strong, given the sample sizes and the leakage proportions.  Moreover, the number of detections remains low overtime.
In fact, the median number of detection is always below 3.

The results for Case II are given in Figures~\ref{fig:sim2_2}. Similar to the results of Algorithm~\ref{alg:one-step-unknown} as in Study I, the FNPs in Case II are smaller than those in Case I, meaning that Algorithm~\ref{alg:one-step-unknown} again makes more conservative decisions. Specifically, the median FNP is always  below the 0.004 level.
The FDP values are acceptable, but are  much larger than those in Case I. Specifically, the median FDP is always below 0.73. Finally,
the numbers of detections are still reasonably low, with the median number of detections always below 7. It would be affordable for the testing program to review detected items when detection size is of this scale.


%

\begin{figure}
  \centering
  \includegraphics[scale = 0.4]{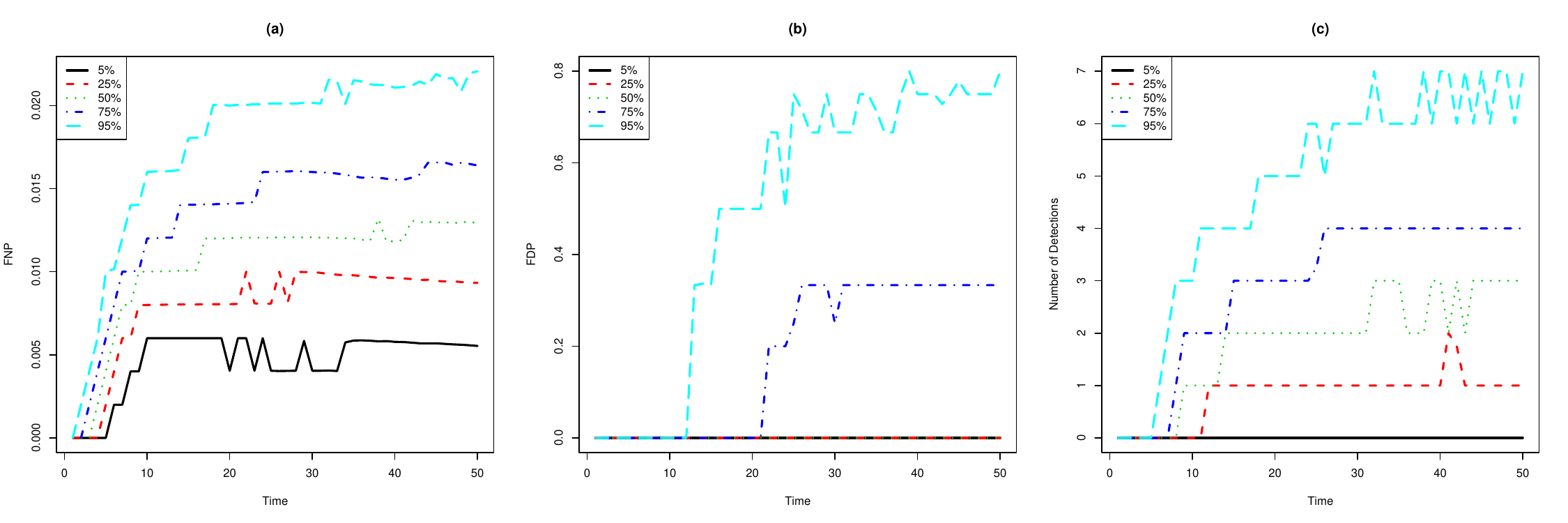}
  \caption{The performance of Algorithm~\ref{alg:one-step} when data are generated from an IRT model
   under a setting that mimics operational tests, and the monitoring statistics are constructed based on item response data.
   The plots can be interpreted the same as those in Figure~\ref{fig:sim1_1}.}\label{fig:sim2_1}
\end{figure}

\begin{figure}
  \centering
  \includegraphics[scale = 0.4]{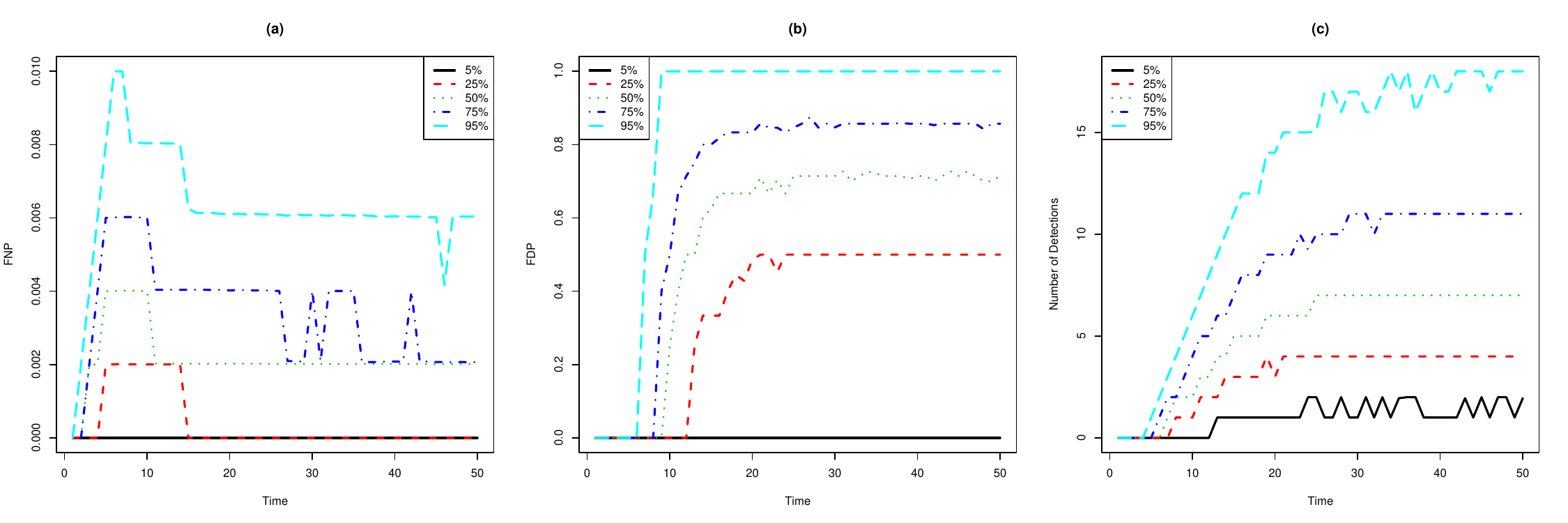}
  \caption{The performance of Algorithm~\ref{alg:one-step-unknown} when data are generated from an IRT model
   under a setting that mimics operational tests, and the monitoring statistics are constructed based on item response data.
   The plots can be interpreted the same as those in Figure~\ref{fig:sim1_1}.}\label{fig:sim2_2}
\end{figure}

\section{Concluding Remarks} \label{sec:conc}

In this paper, we provide a compound change detection framework for sequential item quality  control, one of the most important problems in educational testing. A Bayesian change point model is proposed in both general and specific forms.
Compound decision rules are proposed when the Bayesian change point model is completely known and when only partial information is available about the model. Theoretical properties of these decision rules are established and their empirical performance is evaluated by simulations under various settings. Our simulation studies
show that the proposed method reasonably controls the proportion of post-change items among the undetected ones without making too many false detections of pre-change items, suggesting that the proposed method may be applicable to operational tests for their quality control.
Our simulation results also show that the proposed method is quite robust against model misspecification. More specifically, although our method is developed under a change point model assuming independent monitoring statistics, it still performs well when there
exists weak positive dependence among the monitoring statistics. 

\yc{We clarify that the proposed method controls local FNR but does not control local FDR.
The local FDR in each step of our procedure is a result of the signal of the data and the threshold we use for local FNR. Given a threshold for our local FNR and an IRT setting, there is no simple way to characterize the relationship between the
sample size and the local FDR, as the signal in the data depends on many different factors, not
only the sample size, but also the number of test takers to whom each post-change item is
leaked, population of test takers, the number of items in a test, and the characteristics of the
items in the pool (e.g., the distributions of the discrimination and difficulty parameters). For
a real test, we would suggest to run simulation studies to understand how the local FDR depends
on these factors, given the parameters of the items in the pool and the target level for local FNR.}

{One limitation of the current work lies in the simulation study.
As discussed in Section~\ref{sec:dec}, the proposed method can be generalized to continuous tests, for which
item pool monitoring may be more important than frequent but non-continuous tests. The current simulation only considers settings for frequent but non-continuous tests. Its results do not imply the performance of the proposed method under settings for continuous tests, where the monitoring statistic for each stream is likely obtained from small and varying sample sizes. Simulation studies under real continuous test settings will be conducted in future research.}

Another limitation is that the specific model with independent geometric change points may not be flexible enough. For example, the change points are likely correlated, driven by some events such as the leakage of a set of items to the public or the change of curriculum that may affect multiple items.
A challenge from removing the independent geometric distribution assumption
is that
the posterior probabilities $P(\tau_k < t \vert \mathcal F_t)$ typically do not have an analytic form. Several questions are worth future investigation. First, can we still control the compound risk using a misspecified independent geometric model? Second, can we develop  methods for approximating these posterior probabilities, such as variational approximation and certain Bayesian filters \citep[e.g., Chapter 10,][]{bishop2006pattern}?

In practice, there is always unknown information in the Bayesian change point model, especially the distribution of change points and the post-change distribution. The current solution is to take a conservative approach that makes decision under essentially the worst-case model. This approach guarantees the control of the compound risk  for a finite sample, but there might be a sacrifice in making more false detections of pre-change items.
An alternative is to take an online estimation approach  that estimates the unknown model parameters sequentially together with
the sequential change detection process. \yc{The unknown parameters can be treated as random variables and estimated by a Bayesian approach, or as fixed parameters and estimated by a likelihood-based approach.}
Similar to the multi-armed bandit problem \citep[e.g.,][]{robbins1952some,sutton2018reinforcement}, this approach also faces an exploration–exploitation tradeoff dilemma, i.e., the trade-off between the efforts to achieve a more accurate estimation and to make better decision.
This problem is left for future investigation.


Many sequential change detection applications involve multiple data streams, such as the detection of customer behavior change in e-commence, the detection of changed sensors in engineering, and the detection of abnormal change in stocks.
For such problems, it may be
more sensible to control FDR- or FNR-type compound risks than individual risks for single data streams.
Although this paper focuses on item quality control in educational testing, the proposed methodological framework is very general and applicable to many other multi-stream change detection problems in different fields. In fact, the proposed procedure can be easily modified to control local FDR.

\vspace{2cm}
\appendix

\noindent
{\Large\bf Appendix}

\section{Additional Simulation Results}

In the appendix, we have provided two additional simulation studies. Both studies take the same setting as Study I, except that in the first study (Study A1), we set the threshold for the risk to be 5\% instead of 1\%, and in the second study (Study A2), we sample $\mu_k$ from the uniform distribution over the interval [2,3].
\subsection{Study A1}

The results are given in Figures~\ref{fig:simA1_1} through \ref{fig:simA1_4} that are parallel to Figures \ref{fig:sim1_1} through \ref{fig:sim1_4}, respectively. Note that the only difference between the current study and Study I is that the threshold for the risk is set to be 5\% instead of 1\%.

\begin{figure}
  \centering
  \includegraphics[scale = 0.4]{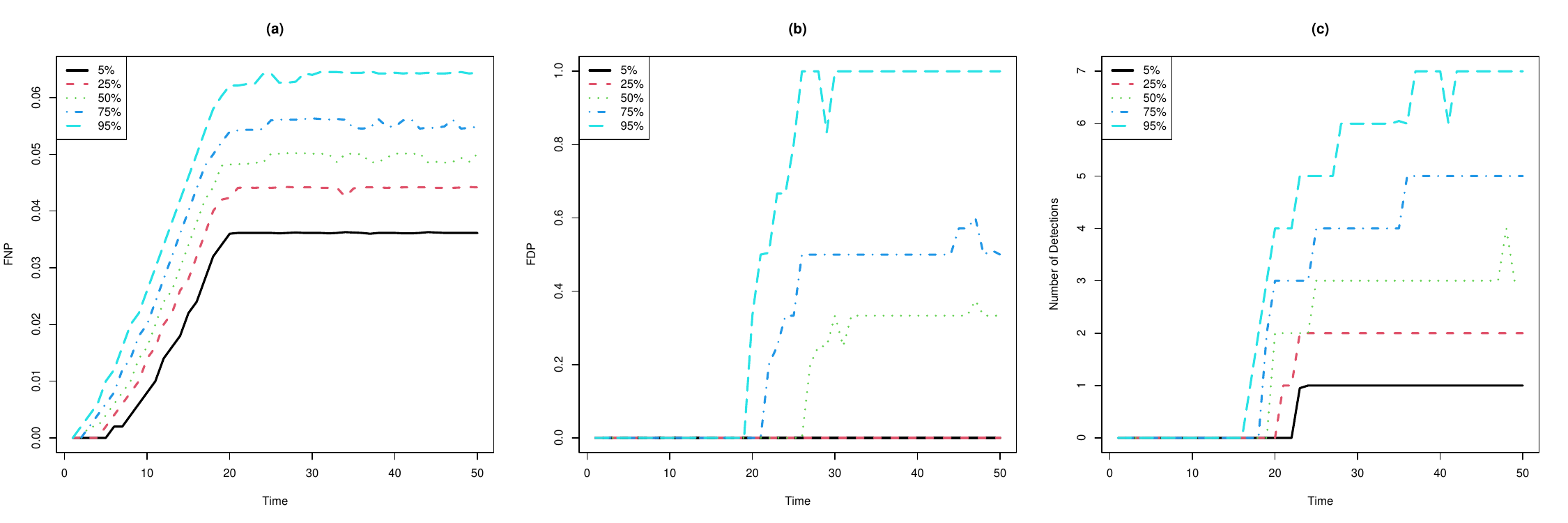}
  \caption{The performance of Algorithm~\ref{alg:one-step} under model correct specification. The plots can be interpreted the same as those in Figure~\ref{fig:sim1_1}.}\label{fig:simA1_1}
\end{figure}

\begin{figure}
  \centering
  \includegraphics[scale = 0.4]{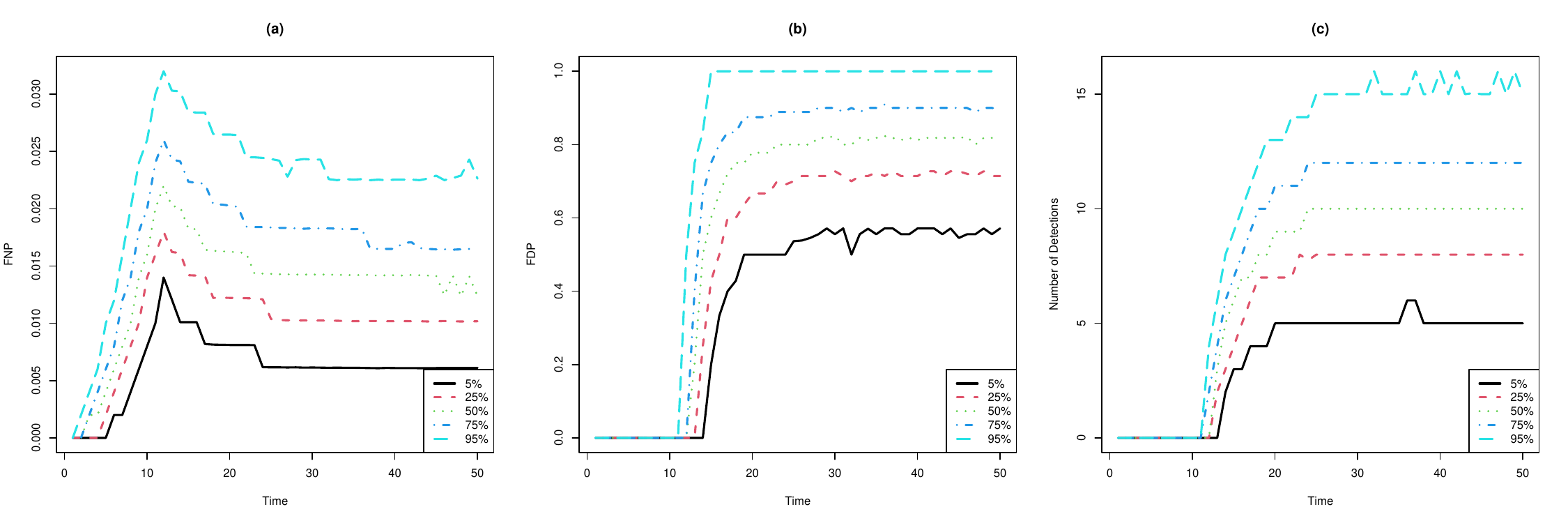}
  \caption{The performance of Algorithm~\ref{alg:one-step-unknown} under model correct specification.  The plots can be interpreted the same as those in Figure~\ref{fig:sim1_1}.}\label{fig:simA1_2}
\end{figure}

\begin{figure}
  \centering
  \includegraphics[scale = 0.4]{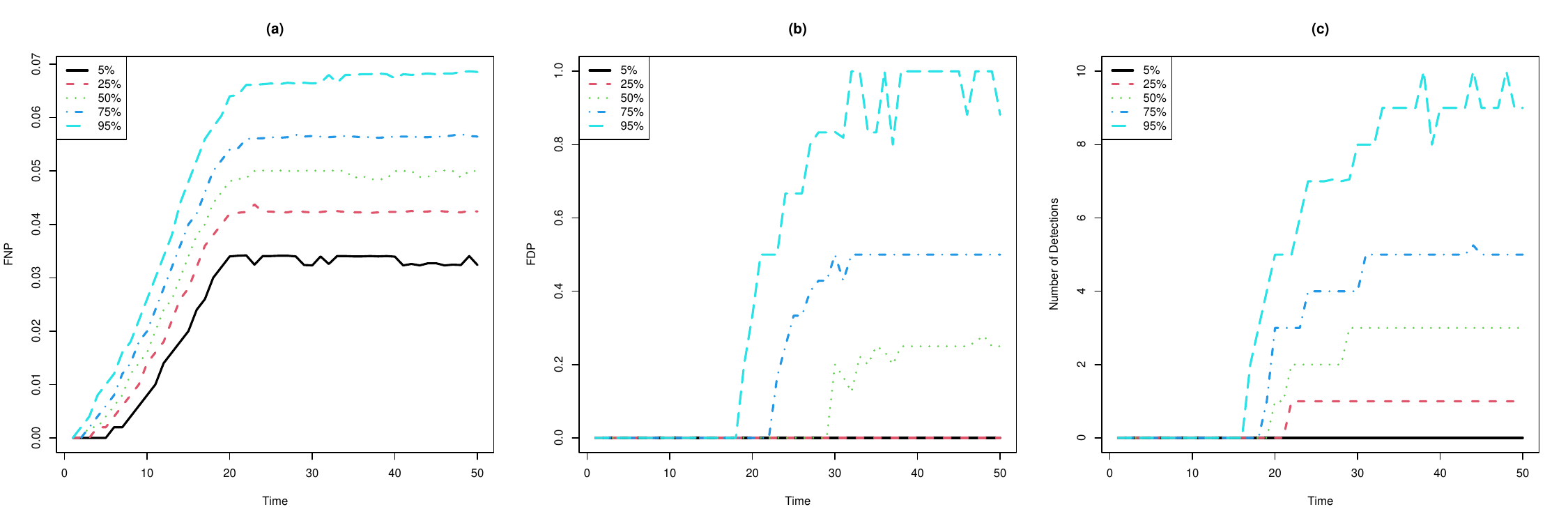}
  \caption{The performance of Algorithm~\ref{alg:one-step} under model mispecification.  The plots can be interpreted the same as those in Figure~\ref{fig:sim1_1}.}\label{fig:simA1_3}
\end{figure}

\begin{figure}
  \centering
  \includegraphics[scale = 0.4]{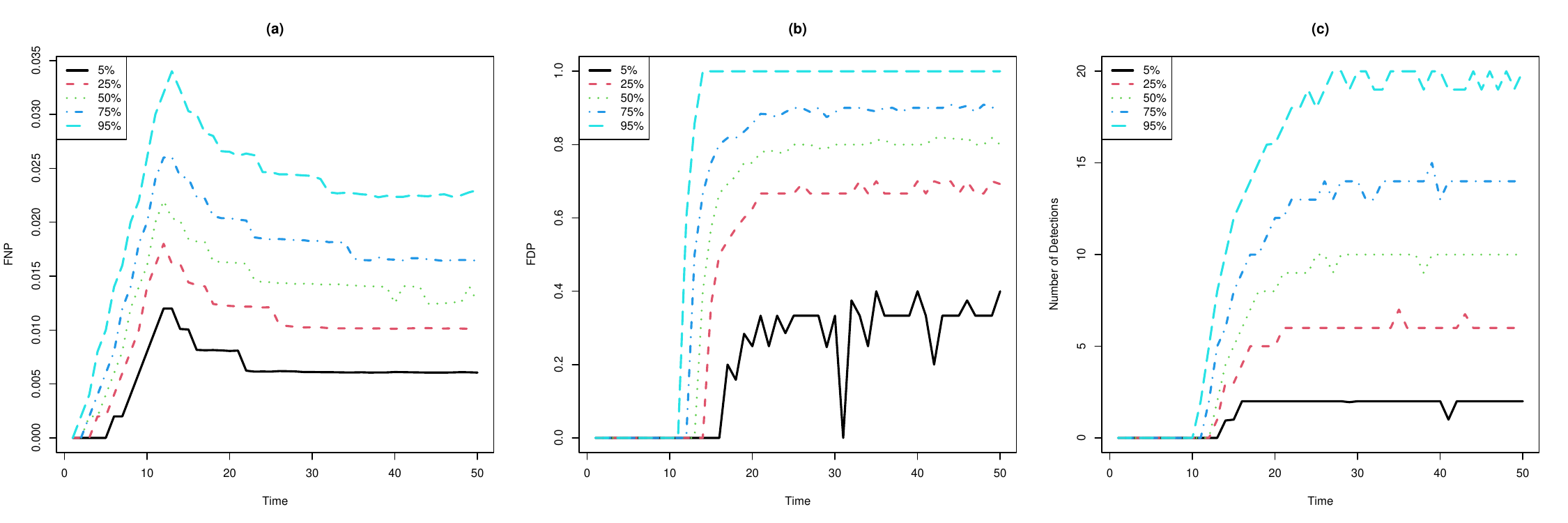}
  \caption{The performance of Algorithm~\ref{alg:one-step-unknown} under model mispecification.  The plots can be interpreted the same as those in Figure~\ref{fig:sim1_1}.}\label{fig:simA1_4}
\end{figure}

\subsection{Study A2}

The results are given in Figures~\ref{fig:simB1_1} through \ref{fig:simB1_4} that are parallel to  Figures \ref{fig:sim1_1} through \ref{fig:sim1_4}, respectively. Note that the only difference between the current study and Study I is that  $\mu_k$ of the post-change distribution is generated from uniform distribution over the interval $[2,3]$ instead of $[1,2]$.

\begin{figure}
  \centering
  \includegraphics[scale = 0.4]{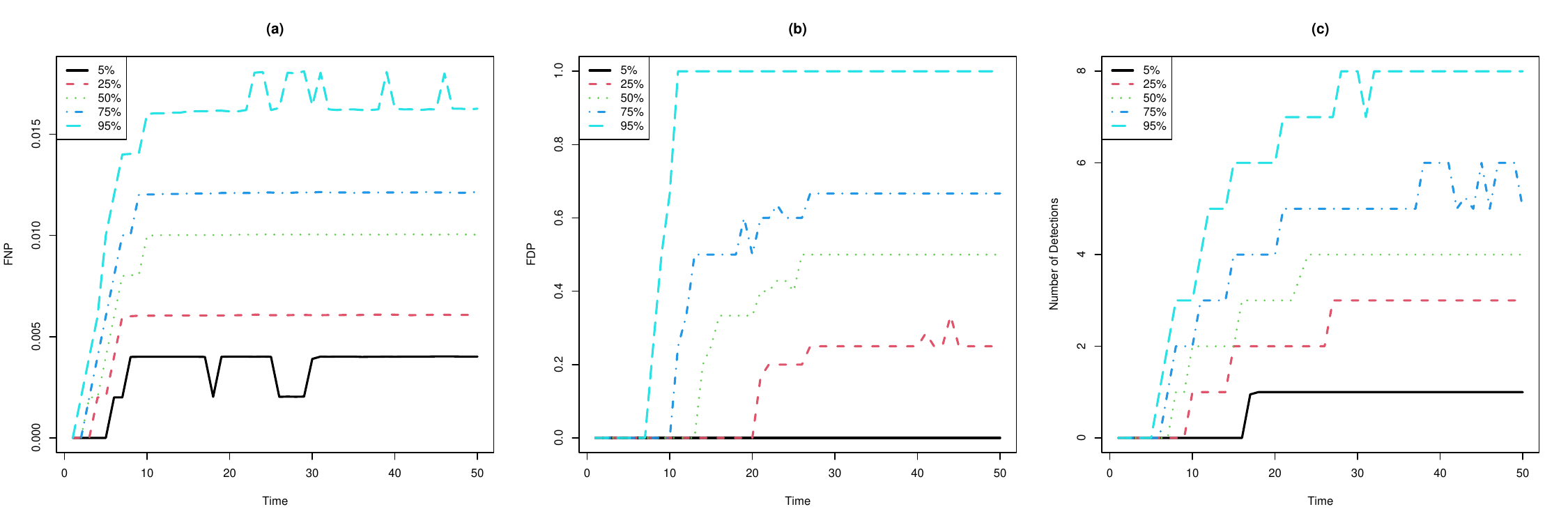}
  \caption{The performance of Algorithm~\ref{alg:one-step} under model correct specification. The plots can be interpreted the same as those in Figure~\ref{fig:sim1_1}.}\label{fig:simB1_1}
\end{figure}

\begin{figure}
  \centering
  \includegraphics[scale = 0.4]{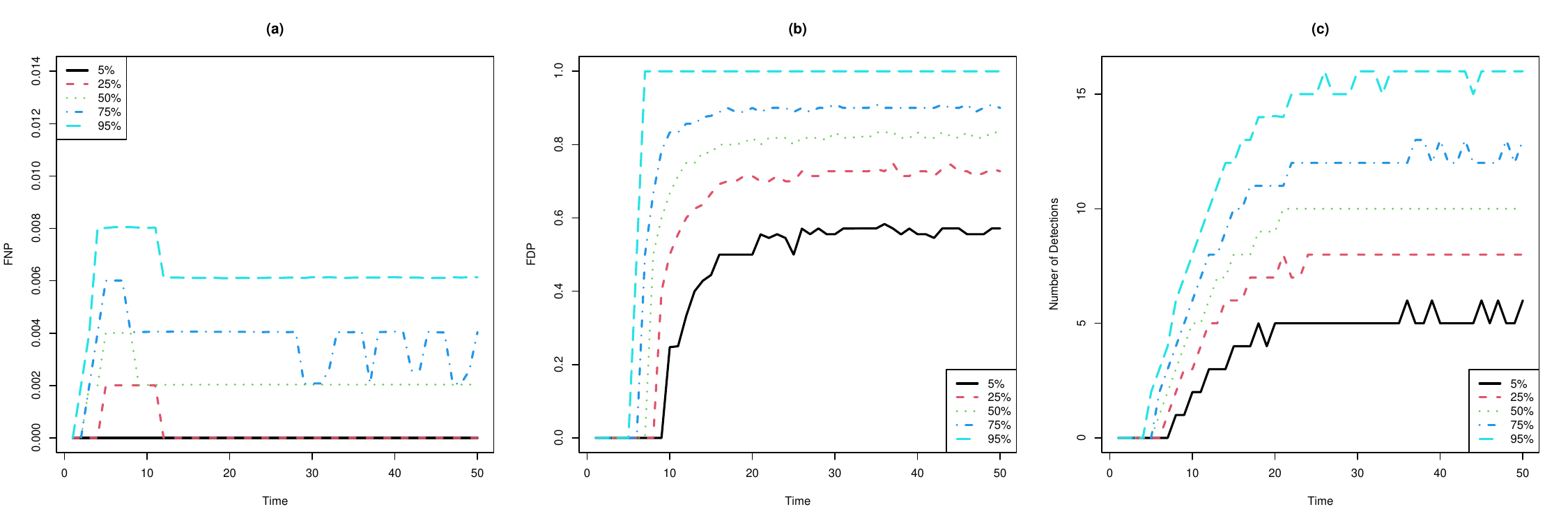}
  \caption{The performance of Algorithm~\ref{alg:one-step-unknown} under model correct specification.  The plots can be interpreted the same as those in Figure~\ref{fig:sim1_1}.}\label{fig:simB1_2}
\end{figure}

\begin{figure}
  \centering
  \includegraphics[scale = 0.4]{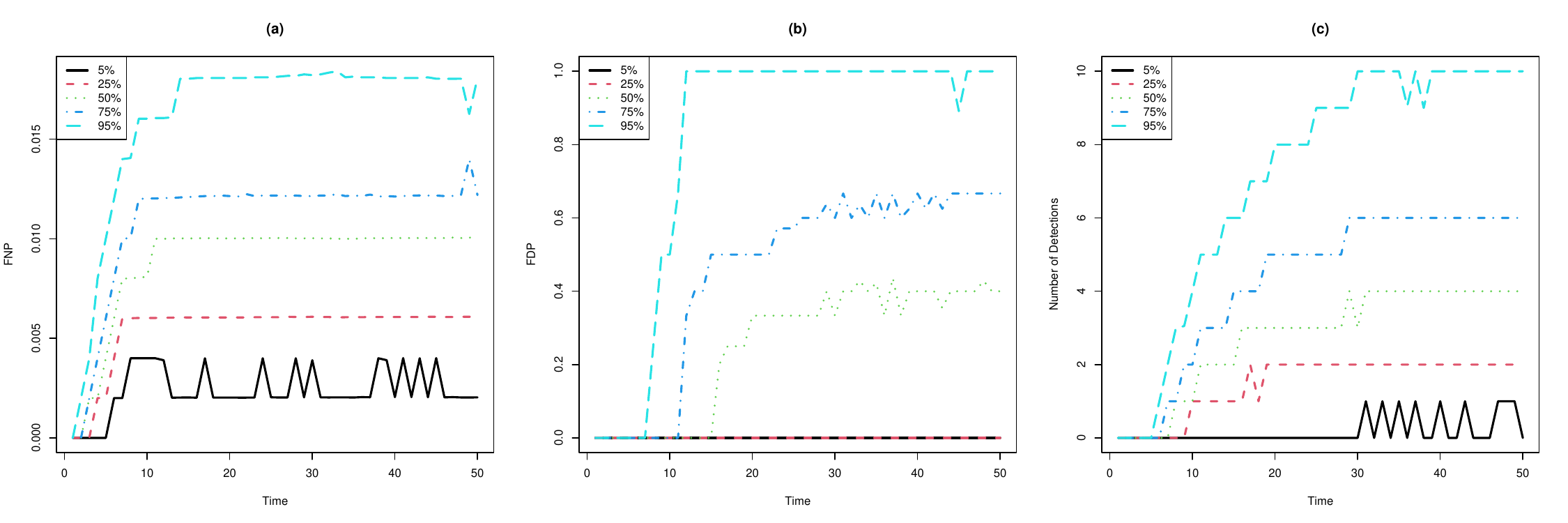}
  \caption{The performance of Algorithm~\ref{alg:one-step} under model mispecification.  The plots can be interpreted the same as those in Figure~\ref{fig:sim1_1}.}\label{fig:simB1_3}
\end{figure}

\begin{figure}
  \centering
  \includegraphics[scale = 0.4]{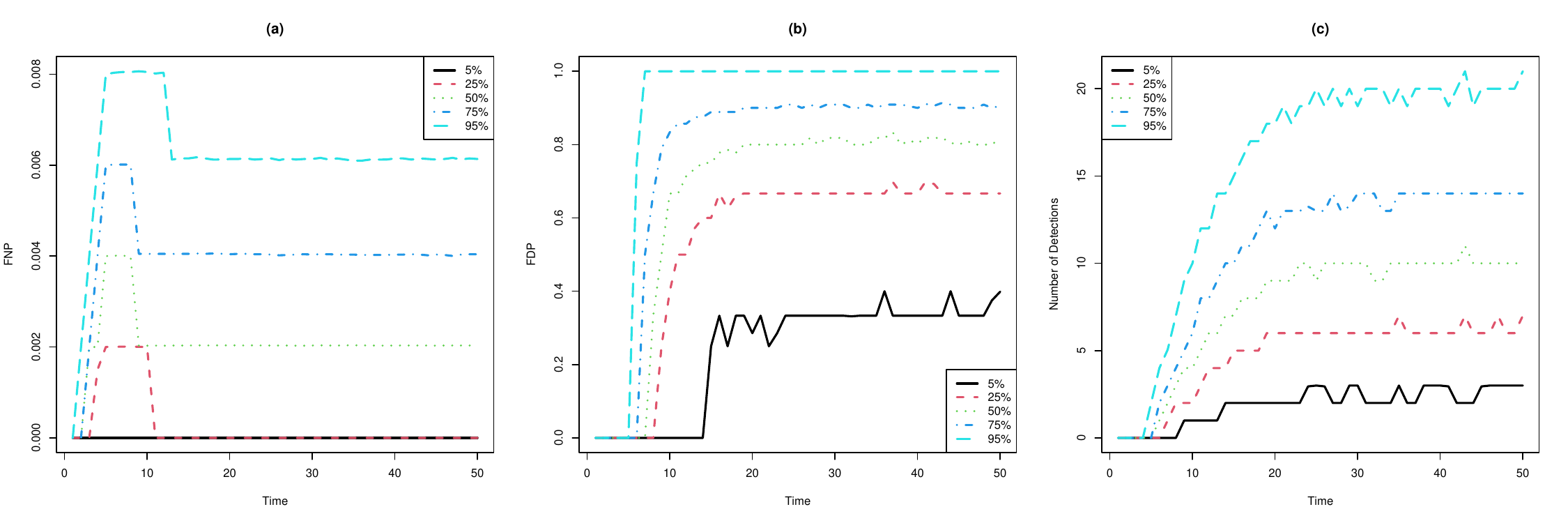}
  \caption{The performance of Algorithm~\ref{alg:one-step-unknown} under model mispecification.  The plots can be interpreted the same as those in Figure~\ref{fig:sim1_1}.}\label{fig:simB1_4}
\end{figure}

\section{Additional Theoretical Results}

The proposed detection rule is not only locally optimal at each time point $t$ in the sense given in Proposition~\ref{prop:one-step}, but also uniformly optimal at all time points for some specific models.
More precisely, consider the following $S_t^*$ obtained by a uniform sampling. For $k\in S_t$, let $1_{\{k\in S_t^*\}}$ be i.i.d. Bernoulli variables with a parameter $\lambda\in(0,1]$ given $S_t$.

\begin{theorem}\label{thm:uniform-optimality}
Assume the set $S_t^*$ is obtained by a uniform sampling described above and no new items are added, {the items are removed once they are detected}, the change points follow the geometric model in \eqref{eq:geo-prior}, and there exists $\rho\in (0,1)$ and density functions $p$ and $q$ such that $\rho_k=\rho$, $p_k=p$ and $q_k=q$ for all $k$.
Then, for any other decision rule $\{D'_t\}_{t\geq 1}$ satisfying {$R(D'_t|\mathcal{F}'_t)  \leq \alpha$ for all $t$,
$E(|D'_t|)\leq E(|D_t|)$ for all $t$. Here, $\{\mathcal{F}'_t\}_{t\geq 1}$ denotes the information filtration associated with $\{D'_t\}_{t\geq 1}$.
}
\end{theorem}
Although the above Theorem~\ref{thm:uniform-optimality} is presented in a relatively simple form, its proof is quite involved, requiring technical tools such as monotone coupling and stochastic ordering over a special non-Euclidean space. The technical derivations are given in Appendix C.
{
\begin{remark}[Conditional risk and unconditional risk]
The optimality property  in Theorem~\ref{thm:uniform-optimality} is established among all the methods controlling the conditional expectation of the risk. In general, these optimality results no longer hold if our goal is instead to control the unconditional risk in each step, because the unconditional risk $E(R(D^*_t|\mathcal F_t))$ may be strictly smaller than $\alpha$. In such cases, it may be possible to detect fewer changes so that the unconditional risk is getting closer to the threshold $\alpha$.   On the other hand, we expect that  $E(R(D^*_t|\mathcal F_t))$ is approaching $\alpha$,  as the number data streams grows and $t$ is large enough. See Theorem 2 in \cite{chen2019compound}  for a rigorous justification of this statement under  additional conditions. Thus, we expect
  the proposed method to be approximately optimal under the unconditional risk as the number of data streams grows to infinity.
\end{remark}
}

\section{Proof of Theoretical Results}
In this section, we provide proof of theoretical results. We will postpone the proof of Theorem~\ref{thm:uniform-optimality} to the end of the section, because it is much more technical than the other theoretical results.

\subsection{Proof of Proposition~\ref{prop:one-step}}
Note that $R(D_t|\mathcal{F}_t) = \frac{\sum_{k\in S_t\setminus D_t} P(\tau_k < t|\mathcal F_t)}{\max\{|S_t \setminus D_t|,1\}} = \frac{\sum_{k\in S_t\setminus D_t} W_{kt}}{\max\{|S_t \setminus D_t|,1\}}$. Comparing $R(D_t|\mathcal{F}_t)$ with steps 1--3 in
 Algorithm~\ref{alg:one-step}, we can see that $ \frac{\sum_{k\in S_t\setminus D_t} W_{kt}}{\max\{|S_t \setminus D_t|,1\}} = V_n\leq \alpha$.
Thus, $R(D_t|\mathcal{F}_t)\leq \alpha$ for all $t$.

Let $D'_t$ be another detection set that is $\mathcal F_t$ measurable and $R(D'_t|\mathcal{F}_t)\leq \alpha$. We first prove that  $|D_t| \leq |D'_t|$. We prove by contradiction.
Recall that according to Algorithm~1, $D_t = S_t \setminus \{k_1, ..., k_n\}$ if $n\geq 1$ and $D_t = S_t$ if $n = 0$, where $n$ is the largest value in $\{0, 1, ..., |S_t|\}$ such that
$V_{n} \leq \alpha.$  Suppose that $|D_t'| < |D_t|$.  Then $|S_t\setminus D_t'| > n$ and
$$\frac{\sum_{k\in S_t\setminus D_t'} W_{kt}}{|S_t \setminus D_t'|} \leq \alpha.$$
It means that there exists $n' > n$ such that $V_{n'} \leq \alpha.$ This contradicts with the construction of $n$.

We now prove
$E\big(\sum_{k\in D_t} 1_{\{\tau_k \geq t\}}\vert \mathcal F_t\big) \leq E\big(\sum_{k\in D'_t} 1_{\{\tau_k \geq t\}} \vert \mathcal F_t\big),$  which further implies $$E\big(\sum_{k\in D_t} 1_{\{\tau_k \geq t\}}\big) \leq E\big(\sum_{k\in D'_t} 1_{\{\tau_k \geq t\}}\big).$$
Note that we only need to prove
$$\sum_{k\in D_t'}(1-W_{kt}) - \sum_{k\in D_t}(1-W_{kt}) = \sum_{k\in D_t'\setminus D_t}(1-W_{kt}) - \sum_{k\in D_t\setminus D_t'}(1-W_{kt})  \geq 0.$$
Since we have proved that $|D_t| \leq |D'_t|$, there are more terms in the summation $\sum_{k\in D_t'\setminus D_t}(1-W_{kt})$ than in $\sum_{k\in D_t\setminus D_t'}(1-W_{kt})$. In addition, by the way $D_t$ is constructed, for any $k \in D_t\setminus D_t'$ and $k' \in D_t'\setminus D_t$, $1-W_{kt} \leq 1-W_{k't}$. Thus, $\sum_{k\in D_t'\setminus D_t}(1-W_{kt}) - \sum_{k\in D_t\setminus D_t'}(1-W_{kt})  \geq 0$ holds. This completes our proof.


\subsection{Proof of Proposition~
\ref{prop:update-post}}
Let $\{l_1,l_2,\cdots\}$ be the indices of time that the item $k$ is used and thus $X_{l_r}$s are obtained. That is, $l_1=\inf\{t:k\in S_t^*\}$, $l_{m+1}=\inf\{t>l_m:k\in S_t^*\}$ for $m=1,2,\cdots$.
It is not hard to verify (e.g., by induction) that $U_{kt}$ satisfy
\begin{equation}\label{eq:ukt-formula}
	U_{kt}= \sum_{s=1}^{e_{kt}-1}\prod_{r=s+1}^{e_{kt}}\big\{\frac{q_{k,l_r}(X_{k,l_r})/p_{k,l_r}(X_{k,l_r})}{1-\rho_k}\big\}.
\end{equation}
On the other hand, the posterior probability is calculated using Bayes formula as follows,
\begin{equation}
	\begin{split}
		W_{kt}&=\frac{\sum_{s=1}^{e_{kt}-1}P(\gamma_k=s)\prod_{r=1}^s p_{k,l_r}(X_{k,l_r})\prod_{r=s+1}^{e_{kt}}q_{k,l_r}(X_{k,l_r})}{\sum_{s=1}^{e_{kt}-1}P(\gamma_k=s)\prod_{r=1}^s p_{k,l_r}(X_{k,l_r})\prod_{r=s+1}^{e_{kt}}q_{k,l_r}(X_{k,l_r})+P(\gamma_k\geq e_{kt}|e_{kt})\prod_{s=1}^{e_{kt}} p_{k,l_s}(X_{l_s})}\\
		&= \frac{\sum_{s=1}^{e_{kt}-1}P(\gamma_k=s)\prod_{r=s+1}^{e_{kt}}\{q_{k,l_r}(X_{k,l_r})/p_{k,l_r}(X_{k,l_r})\}}{\sum_{s=1}^{e_{kt}-1}P(\gamma_k=s)\prod_{r=s+1}^{e_{kt}}\{q_{k,l_r}(X_{k,l_r})/p_{k,l_r}(X_{k,l_r})\}+P(\gamma_k\geq e_{kt}|e_{kt})}.
	\end{split}
\end{equation}
Plug $P(\gamma_k=s)=(1-\rho_k)^{s-1}\rho_k$ and $P(\gamma_k\geq e_{kt}|e_{kt})=(1-\rho_k)^{e_{kt}-1}$ into the above display. Then,
\begin{equation}
\begin{split}
	W_{kt}	&= \frac{\sum_{s=1}^{e_{kt}-1}(1-\rho_k)^{s-1}\rho_k\prod_{r=s+1}^{e_{kt}}\{q_{k,l_r}(X_{k,l_r})/p_{k,l_r}(X_{k,l_r})\}}{\sum_{s=1}^{e_{kt}-1}(1-\rho_k)^{s-1}\rho_k\prod_{r=s+1}^{e_{kt}}\{q_{k,l_r}(X_{k,l_r})/p_{k,l_r}(X_{k,l_r})\}+(1-\rho_k)^{e_{kt}-1}}\\
	& =  \frac{\sum_{s=1}^{e_{kt}-1}(1-\rho_k)^{s-e_{kt}}\prod_{r=s+1}^{e_{kt}}\{q_{k,l_r}(X_{k,l_r})/p_{k,l_r}(X_{k,l_r})\}}{\sum_{s=1}^{e_{kt}-1}(1-\rho_k)^{s-e_{kt}}\prod_{r=s+1}^{e_{kt}}\{q_{k,l_r}(X_{k,l_r})/p_{k,l_r}(X_{k,l_r})\}+1/\rho_{k}},
\end{split}
\end{equation}
which is the same as $\frac{U_{kt}}{U_{kt}+1/\rho_k}$.

\subsection{Proof of Proposition~\ref{prop:w-extended}}
According to \eqref{eq:ukt-formula}, $U_{kt}(\rho_k,\boldsymbol\pi)$ is increasing in $\rho_k$ while fixing $\boldsymbol\pi$. Thus, $W_{kt}(\rho_k,\boldsymbol\pi)=\frac{U_{kt}(\rho_k,\boldsymbol\pi)}{U_{kt}(\rho_k,\boldsymbol\pi)+1/\rho_k}$ is also increasing in $\rho_k$. This implies $\sup_{\boldsymbol\pi\in\Theta}W(\overline{\rho},\boldsymbol\pi) =\sup_{\rho\in[0,\overline{\rho}],\boldsymbol\pi\in\Theta}W(\rho,\boldsymbol\pi)=\overline{W}_{kt}$.
By Proposition~\ref{prop:update-post} and the definition of $\overline{U}_{kt}$, we have $\frac{\overline{U}_{kt}}{\overline{U}_{kt}+1/\overline{\rho}}=\sup_{\boldsymbol\pi\in\Theta}W(\overline{\rho},\boldsymbol\pi)$. Thus, $\overline{W}_{kt}=\sup_{\boldsymbol\pi\in\Theta}W(\overline{\rho},\boldsymbol\pi) =\frac{\overline{U}_{kt}}{\overline{U}_{kt}+1/\overline{\rho}}$.

\subsection{Proof of Theorem~\ref{thm:unknown}}
Note that by definition $R(D_t|\mathcal{F}_t) = \frac{\sum_{k\in S_t\setminus D_t} P(\tau_k < t|\mathcal F_t)}{\max\{|S_t \setminus D_t|,1\}} = \frac{\sum_{k\in S_t\setminus D_t} W_{kt}(\rho_k,\boldsymbol\pi)}{\max\{|S_t \setminus D_t|,1\}}\leq \frac{\sum_{k\in S_t\setminus D_t} \overline{W}_{kt}}{\max\{|S_t \setminus D_t|,1\}}$. The rest of the proof follows similarly to the proof of Proposition~\ref{prop:one-step}.

\subsection{Proof of Theorem~\ref{thm:uniform-optimality}}
We first note that Theorem~\ref{thm:uniform-optimality} under the current settings is an extension of  Theorem~1 in \cite{chen2019compound}. Indeed, Theorem~1 in \cite{chen2019compound} is proved under the model where $S_t^*=S_t$ for all $t$, which corresponds to the uniform sampling with the sampling weight $\lambda=1$. For the current settings, we will follow a similar proof strategy as that of Theorem~1 in \cite{chen2019compound}.
Because the complete proof is  lengthy and quite technical, we will only emphasize the key differences and omit similar details.

The current settings and that of \cite[Theorem 1 and Theorem 5]{chen2019compound} are mainly different at two places. First, the information filtration $\mathcal{F}_t$ is richer under the current setting because of the additional information in $\{S^*_s\}_{s=1,\cdots, t}$. Second, the conditional distribution of $W_{k,t+1}$ given $W_{k,t}$ for $k\in S_{t}$ is different because of the possible difference in $S^*_t$ and $S_t$ and the randomness in the uniform sampling. These differences  affects how we extend the supporting Lemmas \cite[Lemmas D.1 -- D.10]{chen2019compound} to the current settings. We first note that Lemmas D.1 -- D.6 in \cite{chen2019compound} either directly apply to the current settings or can be easily extended as they are not related to the random sampling $\{S^*_{s}\}_{1\leq s\leq t}$. Lemmas D.7 -- D.9 in \cite{chen2019compound} are not well-defined under the current settings because of the random sampling of $\{S^*_{s}\}_{1\leq s\leq t}$ and the missing data induced by it. Nevertheless, by modifying the proof, Lemma D.7  -- D.9 in \cite{chen2019compound} can be replaced and extended to the current settings with the next lemma.

\begin{lemma}[Modified Lemma~D.9 of \cite{chen2019compound}]\label{lemma:modified-d9}
	{Let $\{V_{kt}\}_{t\geq 1}$ be defined the same as $W_{kt}$ with no item removed. That is, $V_{kt}=\frac{U_{kt}}{U_{kt}+1/\rho_k}$ and $\{U_{kt}\}$ evolves according to \eqref{eq:update} by setting $S_t=S_1$ for all $t$.}
Then, $\{V_{kt}\}_{t\geq 1}$ is a homogeneous Markov chain, in addition,  its transition kernel is stochastically monotone. We will later refer to this transition kernel as $K(\cdot,\cdot)$.
\end{lemma}

\begin{proof}[Proof of Lemma~\ref{lemma:modified-d9}]
	We first derive the conditional distribution of $(V_{k,t+1},1_{\{k\in S^*_{t+1}\} })$ given $V_{k,1},\cdots,V_{k,t}$. According to Proposition~\ref{prop:update-post} and that $S^*_{t+1}$ is obtained by a uniform sampling with a weight $\lambda$, we have $1_{  \{ k\in S^*_{t+1}\}  }\sim\text{Bernoulli}(\lambda)$, independent with $V_{k,1},\cdots,V_{k,t}$, and
	\begin{equation}\label{eq:v-on-u}
	V_{k,t+1} =
		\begin{cases}
			V_{k,t} &\text{ if }  1_{\{k\in S^*_{t+1}\}} = 0\\
			\frac{(U_{k,t}+1)\frac{L_{k,t+1}}{1-\rho}}{(U_{k,t}+1)\frac{L_{k,t+1}}{1-\rho}+\frac{1}{\rho}} &\text{ if }  1_{\{k\in S^*_{t+1}\}} = 1
		\end{cases},
	\end{equation}
	where $L_{k,t+1}:=q(X_{k,t+1})/p(X_{k,t+1})$.
	In addition, given $1_{\{k\in S^{*}_1\}},\cdots 1_{\{k\in S^*_{t}\}}$, $X_{k,s}$ for $k\in S^*_{s}$ ($s=1,\cdots,t$), and $1_{\{k\in S^*_{t+1}\}} = 1$, the conditional density of $X_{k,t+1}$ is $\delta_{k,t}q(x)+(1-\delta_{k,t})p(x)$, where
	\begin{equation}
		\begin{split}
		\delta_{k,t}:&= P(\tau_k\leq t|1_{\{k\in S^{*}_s\}} \text{ and }X_{k,s} \text{ if } k\in S^*_s, \text{ for }1\leq s\leq t )\\
		&=\frac{\sum_{s=1}^{e_{kt}-1}P(\gamma_k=s)\prod_{r=1}^s p(X_{k,l_r})\prod_{r=s+1}^{e_{kt}}q(X_{k,l_r}) + P(\gamma_k=e_{kt}|e_{kt})\prod_{r=1}^{e_{kt}} p(X_{k,l_r})}{\sum_{s=1}^{e_{kt}-1}P(\gamma_k=s)\prod_{r=1}^s p(X_{k,l_r})\prod_{r=s+1}^{e_{kt}}q(X_{k,l_r})+P(\gamma_k\geq e_{kt}|e_{kt})\prod_{s=1}^{e_{kt}} p(X_{l_s})}
		\end{split}
	\end{equation}
	Simplifying the above display in a same way as we did in the proof of Proposition~\ref{prop:update-post}, we arrive at
	\begin{equation}\label{eq:delta-on-u}
		\delta_{k,t} = \frac{U_{k,t}+1}{U_{k,t}+1/\rho}.
	\end{equation}
	From the above derivation, we can see that the conditional distribution of $V_{k,t+1}$ given $V_{k,1}\cdots, V_{k,t}$ is determined by $V_{k,t}$ in the same way across time. Thus, it is a homogeneous Markov chain. To see it is stochastically monotone, we construct coupling a $(\hat{V},\hat{V}')$ for the conditional distribution of $V_{k,t+1}$ given $V_{k,t}=v$ and $V_{k,t}=v'$ with $v\leq v'$ as follows.

	First, let $\hat{U}=\frac{v}{\rho(1-v)}$, $\hat{U}'=\frac{v'}{\rho(1-v')}$, $\hat{\delta}=\frac{\hat{U}+1}{\hat{U}+1/\rho}$, and $\hat{\delta}'=\frac{\hat{U}'+1}{\hat{U}'+1/\rho}$. Because $v\leq v'$, we can easily see that  $\hat{\delta}'\geq \hat{\delta}$ and $\hat{U}'\geq \hat{U}$.
	Second, according to Lemma D.6 in \cite{chen2019compound} and $\hat{\delta}\leq \hat{\delta}'$, there exists a coupling $(\hat{L},\hat{L}')$ such that $\hat{L}$ has the same distribution as $q(X)/p(X)$ where $X\sim \hat{\delta}q(x)+(1-\hat{\delta})p(x)$, $\hat{L}'$ has the same distribution as $q(X')/p(X')$ where $X'\sim \hat{\delta}'q(x)+(1-\hat{\delta}')p(x)$, and $\hat{L}\leq\hat{L}'$ a.s.
	Third, let $Z\sim \text{Bernoulli}(\lambda)$, and construct $(\hat{V},\hat{V}')$ as
	\begin{equation}
		\hat{V}=
		\begin{cases}
			v &\text{ if } Z=0\\
			\frac{(\hat{U}+1)\hat{L}/(1-\rho)}{(\hat{U}+1)\hat{L}/(1-\rho)+1/\rho} &\text{ if }Z=1
		\end{cases} \text{ and }\hat{V}'=
		\begin{cases}
			v' &\text{ if } Z=0\\
			\frac{(\hat{U}'+1)\hat{L}/(1-\rho)}{(\hat{U}'+1)\hat{L}/(1-\rho)+1/\rho} &\text{ if }Z=1
		\end{cases}.
	\end{equation}
	Because $v\leq v'$, $\hat{U}\leq\hat{U}'$, $\hat{L}\leq\hat{L}'$ a.s. and the same $Z$ is used in constructing $\hat{V}$ and $\hat{V}'$, we can see that $\hat{V}\leq\hat{V}$ a.s.

	On the other hand, according to equations \eqref{eq:v-on-u} and \eqref{eq:delta-on-u}, $U_{k,t}=\frac{V_{k,t}}{\rho(1-V_{k,t})}$, and the conditional distribution of $X_{k,t+1}$ given $V_{k,t}$, it is not hard to verify that $\hat{V}$ has the same distribution as $V_{k,t+1}|V_{k,t}=v$ and $\hat{V}'$ has the same  distribution as $V_{k,t+1}|V_{k,t}=v'$. Because such a coupling exists for all $v\leq v'$, we conclude that the transition Kernel is stochastically monotone.
\end{proof}

Next, we extend Lemmas~D.10 -- D.14 in \cite{chen2019compound} to the current settings. We first  define several useful notation. Let $\fd=(d_1,d_2,\cdots,d_t,\cdots)$ denote a detection rule at all the time points. Here, each $d_t$ can be understood as a function that decides $S_{t+1}$ (i.e., selecting $D_t$) based on all historical information up to time $t$ (i.e., a measurable function with respect to $\mathcal F_t$). Also, let $d_1^*,d_2^*,\cdots$ be the one-step detection rules following Algorithm~\ref{alg:one-step} at different time points. In other words, the proposed detection rule is represented as $\fd^*:=(d_1^*,d_2^*,\cdots)$.

Let $\HH_t=\big\{\{X_{ks}\}_{k\in S^*_s}, S^*_{s}, S_s, s=1,\cdots, t \big\}$ denote the historical information obtained up to time $t$ and $\hh_t$ be a realization of $\HH_t$ in its support. In addition, we use $\HH_t^{\fd}$  and $W_{k,t}^{\fd}$ to indicate the corresponding historical information and posterior probability of a change point following the detection rule $\fd$.

Define a partially ordered space $(\spaceo,\lleq)$ as follows.
 Let
\begin{equation}
	 \spaceo=\bigcup_{k=1}^K \left\{\vv=(v_1,\cdots,v_k)\in[0,1]^k: 0\leq v_1\leq \cdots v_k\leq 1
		\right\}\cup \{\ab\},
\end{equation}
where $\ab$ represents a vector with zero length.
For $\uu\in\spaceo$, let $\card(\uu)$ be the length of the vector $\uu$.
A partial order relation $\lleq$ over $\spaceo$ is defined as follows.
For $\uu,\vv\in\spaceo$, we say $\uu\lleq
 \vv$ if $\card(\uu)\geq \card(\vv)$ and $u_i\leq v_i$ for $i=1,...,\card(\vv)$. In addition, we say $\uu\lleq \ab$ for any $\uu\in\spaceo$.
The next lemma is an extension of \cite[Lemma~D.10]{chen2019compound} to our current settings.

\begin{lemma}[Modified Lemmas~D.10 of \cite{chen2019compound}]\label{lemma:modified-d10}
	For any $t\geq 1$ and any decision rule $\fd$, $[W^{\fd}_{S_{t+1}^{\fd},t+1}]$ is conditionally independent with $\mathcal{F}^{\fd}_{t}$ given $[W^{\fd}_{S^{\fd}_{t+1},t}]$. Moreover, the conditional density of $[W^{\fd}_{S^{\fd}_{t+1},t+1}]$ at $\vv$ given $[W^{\fd}_{S^{\fd}_{t+1},t}]=\uu\in\spaceo$ is
	\begin{equation}
	\KKa(\uu,\vv)
	:=
	\begin{cases}
		\sum_{\pi \in \PP_{\card(\uu)}}\prod_{l=1}^{\card(\uu)}K(u_l,v_{\pi(l)}) &\text{ if } \card(\uu)=\card(\vv)\geq 1,\\
		1 & \text{ if } \card(\uu)=\card(\vv)=0,\\
		0 &\text{ otherwise,}
	\end{cases}
	\end{equation}
	where $\PP_{m}$ denotes the set of all permutations over $\{1,\cdots, m\}$ and the transition Kernel $K(\cdot,\cdot)$ is defined in Lemma~\ref{lemma:modified-d9}.
\end{lemma}
\begin{proof}[Proof of Lemma~\ref{lemma:modified-d10}]
The trivial cases where $\dim(\uu)=0$ or $\dim(\uu)\neq\dim(\vv)$ are proved in the same way as that of Lemma D.10 in \ref{lemma:modified-d10}. In the rest of the proof, we will focus on the non-trivial case that $\dim(\uu)=\dim(\vv)=m$ for some positive integer $m$.

First, we consider the conditional distribution of $[W^{\fd}_{S^{\fd}_{t+1},t+1}]$ at $\vv\in\spaceo$, given $\HH^{\fd}_{t}=\hh_t$, $S^{\fd}_{t+1}=s_{t+1}$ and $W^{\fd}_{S_{t+1}^{\fd},t}=\uu$ for $\uu\in\spaceo$,
  Note that now we have $\{S^{*\fd}_s\}_{1\leq s\leq t}$ in the history information $\HH^{\fd}_t$ and thus are conditioning on more variables under the current settings, compared with that of \cite{chen2019compound}. On the other hand, regardless of the information in $\{S^{*\fd}_s\}_{1\leq s\leq t}$, we still have $W^{\fd}_{k,t+1}$s are conditionally independent for different $k\in s_{t+1}$ given $S^{\fd}_{t+1}=s_{t+1}$ and $W^{\fd}_{S^{\fd}_{t+1},t}=\uu$. In addition, given $S^{\fd}_{t+1}=s_{t+1}$ and $W^{\fd}_{S^{\fd}_{t+1},t}=\uu$, $W^{\fd}_{k,t+1}$ is the same as $V_{k,t+1}$ for $k\in s_{t+1}$ and is independent of
$\HH^{\fd}_t=\hh_t$
   with the conditional distribution $\prod_{l=1}^m K(u_l,v_l)$. The rest of the proof is the derivation of the distribution of the order statistic of $W^{\fd}_{S^{\fd}_{t+1},t}$, and is the same as that of Lemma D.10 in \cite{chen2019compound}. Thus, we omit the repetitive details.	
\end{proof}
Following this lemma,  Lemmas D.11--D.14 in \cite{chen2019compound} can be extended to the current settings because their proof are based on Lemmas D.1--D6, D.9--D.10 in \cite{chen2019compound}, which have already been extended to the current settings. With these extensions, we are able to modify and extend key results in \cite{chen2019compound} as follows.

 For each $t_0\in \Zp$, and an arbitrary $\fd$ that controls the risk at a desired level $\alpha$, we consider the following two detection rules based on $\fd$
 $$
\fd_1 = (d_1,\cdots, d_{t_0-1},d^*_{t_0}, d_{t_0+1}^*,\cdots)
\text{ and } \fd_2 = (d_1,\cdots, d_{t_0}, d^*_{t_0+1}, d^*_{t_0+2}\cdots).
 $$
 That is, $\fd_1$ is the detection rule by first select $S_1,\cdots S_{t_0}$ following $\fd$ and switch to the proposed one-step update rule in Algorithm~\ref{alg:one-step} to select $S_{t_0+1},S_{t_0+2},\cdots$. $\fd_2$ is a similar switching rule but switch to the one-step update rule at time $t_0+1$ rather than $t_0$. Note that the $S_t$ selected by $\fd_1$ and $\fd_2$ coincides for $t=1,\cdots, t_0$ and may be different for $t=t_0+1,\cdots.$

\begin{lemma}[Modification of Proposition 4 in \cite{chen2019compound}]\label{lemma:compare-optimal-and-other}
	There exists a coupling of $\spaceo$-valued random variables $(\hat{W},\hat{W}')$ such that
	\begin{equation}
	\hat{W}\eqd \big[ (W_{k,t_0+1}^{\fd_1})_{k\in S_{t_0+1}^{\fd_1}}\big]\big|\HH_{t_0}^{\fd_1} = \hh_{t_0} ~~\text{ and }~~	\hat{W}'\eqd \big[ (W_{k,t_0+1}^{\fd_2})_{k\in S_{t_0+1}^{\fd_2}}\big]\big|\HH_{t_0}^{\fd_2} = \hh_{t_0}
	\end{equation}
	and $\hat{W}\lleq \hat{W}'$ a.s., where $[\cdot]$ denotes the order statistic of a vector and $\eqd$ denotes that random variables on both sides are identically distributed.
\end{lemma}
In the above lemma, we note that $\HH_{t_0}^{\fd_1}$ is the same as $\HH_{t_0}^{\fd_2}$ and further the same as $\HH_{t_0}^{\fd}$ by construction, and thus they have the same support. Also, $(W_{k,t_0+1}^{\fd_2})_{k\in S_{t_0+1}^{\fd_2}}$ is same as $(W_{k,t_0+1}^{\fd})_{k\in S_{t_0+1}^{\fd}}$.

\begin{proof}[Proof of Lemma~\ref{lemma:compare-optimal-and-other}]

The main difference between  Lemma~\ref{lemma:compare-optimal-and-other} and Proposition 4 in \cite{chen2019compound}  is that the history process $\HH^{\fd_l}_{t_0}$ contains additional information about $\{S^{*\fd_l}_s\}_{1\leq s\leq t_0}$ in the current setting ($l=1,2$).
On the other hand, we still have that $W^{\fd}_{S^{\fd}_{t_0},t_0}$ is determined by $\HH^{\fd}_{t_0}=\hh_{t_0}$ and does not depend on $\fd$. We denote it by $\ww_{t_0}$. In addition, $[W^{\fd}_{S^{\fd}_{t_0+1},t_0}]$ is determined by $\hh_{t_0}$ and $\fd$ (through $d_{t_0+1}$). Let us denote $[W^{\fd}_{S^{\fd}_{t_0+1},t_0}]=\ww_{t_0+1}^*$, which is a deterministic function of $\hh_{t_0}$ and $\fd$. According to Lemma~\ref{lemma:modified-d10} and replacing $t$ by $t_0$, the conditional distribution of $[W^{\fd}_{S^{\fd}_{t_0+1},t_0+1}]$ given $\HH^{\fd}_{t_0}=\hh_{t_0}$ has the density function $\KKa(\ww^*_{t_0+1},\cdot)$. Similarly, by replacing $\fd$ with $\fd_1$, we have $[W^{\fd_1}_{S^{\fd_1}_{t_0+1},t_0}]$, denoted by $\ww_{t_0+1}$, is determined by $\hh_{t_0}$ given $\HH^{\fd_1}_{t_0}=\hh_{t_0}$. In addition, the conditional density of $[W^{\fd_1}_{S^{\fd_1}_{t_0+1},t_0+1}]$ is $\KKa(\ww_{t_0+1},\cdot)$.

Note that under the detection rule in Algorithm~\ref{alg:one-step}, Lemma~D.3 and D.4 in \cite{chen2019compound} still hold and thus $\ww_{t_0+1}\lleq \ww_{t_0+1}^*$. The rest of the proof follows similarly as that of Proposition 4 in \cite{chen2019compound}. We omit the repetitive details.
\end{proof}
\begin{lemma}[Modification of Proposition 5 in \cite{chen2019compound}]\label{lemma:optimal-mono}
Let $Y_s=\big[(W_{k,t_0+s}^{\fd_1})_{k\in S_{t_0+s}^{\fd_1}}\big]$. Then, for any $\yy,\yy'\in \spaceo$ such that $\yy\lleq\yy'$, there exists a coupling process $\{(\hat{Y}_s,\hat{Y}'_s)\}_{s=0,1,\cdots}$, satisfying
\begin{enumerate}
 	\item $\{\hat{Y}_s\}_{s\geq 0}\eqd \{Y_s\}_{s\geq 0}\big|Y_0=\yy$, $\{\hat{Y}'_s\}_{s\geq 0}\eqd \{Y_s\}_{s\geq 0}\big|Y_0=\yy'$.
 	\item $\hat{Y}_s\lleq\hat{Y}'_s$ a.s. for all $s\geq 0$.
 \end{enumerate}
 Moreover, the process $(\hat{Y}_s,\hat{Y}'_s)$ does not dependent on $\fd$, $t_0$, nor $\HH_{t_0}^{\fd}$.
\end{lemma}
\begin{proof}[Proof of Lemma~\ref{lemma:optimal-mono}]
	The proof of Lemma~\ref{lemma:optimal-mono} is similar to that of Proposition in \cite{chen2019compound}, which is derived based on Lemmas D.5, D.13 and D.14 in \cite{chen2019compound}. Because we have extended these lemmas to the current settings, the proof of Lemma~\ref{lemma:optimal-mono} follows similarly.
\end{proof}
With Lemmas~\ref{lemma:compare-optimal-and-other} and Lemma~\ref{lemma:optimal-mono}, the rest of the proof follows similarly as that of \cite[Theorem 1 and Theorem 5]{chen2019compound}. We omit the details.




\section{Construction of SIR Statistic}\label{app:SIR}

We denote $$\xi_{kt}^0(m) = \int f(\theta \vert \bbb_k) \frac{1}{\sqrt{2\pi}} \exp(-(\theta - m)^2/2) d\theta.$$
as a function of population mean $m$. Let $\hat m_t$ be given by solving optimization~\eqref{eq:mle}.  Then our estimate of $\xi_{kt}^0$ is
$$\hat \xi_{kt}^0 = \xi_{kt}^0(\hat m_t).$$

We do Taylor expansion of $\xi_{kt}^0(\hat m_t)$ at $m_t^*$ and obtain
$$\hat{\xi}_{kt}^0= \xi_{kt}^0(m_t^*) + \xi_{kt}^{0'}(m_t^*)(\hat m_t - m_t^*) + o_p(1/\sqrt{N_t}),$$
 where $\xi_{kt}^{0'}(m)$ is the derivative of $\xi_{kt}^{0}(m)$ that takes the form
$$\xi_{kt}^{0'}(m) = \int f(\theta \vert \bbb_k) \frac{1}{\sqrt{2\pi}} \exp(-(\theta - m)^2/2) (\theta - m) d\theta.$$
By expanding the score equation for $\hat m_t$ at $m_t^*$, we obtain
$$\hat m_t -  m_t^* = \frac{\sum_{n=1}^{N_t} E(\theta_{tn} - m_t^* \vert  Y_{ktn}, k\in S_t^\dagger )/\kappa_t}{N_t} + o_p(1/\sqrt{N_t}),$$
where $\kappa_t = E[E(\theta_{t1} - m_t^* \vert Y_{kt1}, k\in S_t^\dagger)]^2 = Var(E(\theta_{t1}\vert Y_{kt1}, k\in S_t^\dagger))$.
Therefore,
$$\bar{Y}_{kt} - \xi_{kt}^0(\hat m_t) = \frac{\sum_{n=1}^{N_t} (Y_{ktn} - \xi_{kt}^{0'}(m_t^*) E(\theta_{tn} - m_t^* \vert  Y_{ktn}, k\in S_t^\dagger )/\kappa_t) }{N_t} - \xi_{kt}^0(m_t^*) + o_p(1/\sqrt{N_t}).$$
We thus have
$$Var(\bar{Y}_{kt} - \xi_{kt}^0(\hat m_t)) = \frac{Var(Y_{kt1} - \xi_{kt}^{0'}(m_t^*)E(\theta_{t1} \vert Y_{kt1}, k \in S_t^\dagger )/\kappa_t )}{N_t} + o(1/N_t),$$
which can be approximated by
$$\frac{\sum_{n=1}^{N_t} (Y_{ktn} + \xi_{kt}^{0'}(\hat m_t) \bar\theta_{tn}/\hat\kappa_t - \bar{Y}_{kt} - \xi_{kt}^{0'}(\hat m_t)\bar \theta_t/\hat\kappa_t)^2}{N_t^2},$$
where $$\bar\theta_{tn} = \frac{\int \theta \prod_{k\in S_t^\dagger} f(\theta \vert \bbb_k)^{Y_{ktn}} (1-f(\theta \vert \bbb_k))^{1-Y_{ktn}} \frac{1}{\sqrt{2\pi}} \exp(-(\theta - \hat m_t)^2/2) d\theta}{\int  \prod_{k\in S_t^\dagger} f(\theta \vert \bbb_k)^{Y_{ktn}} (1-f(\theta \vert \bbb_k))^{1-Y_{ktn}} \frac{1}{\sqrt{2\pi}} \exp(-(\theta - \hat m_t)^2/2) d\theta},$$
and
$$\bar \theta_t = \frac{\sum_{n=1}^{N_t} \bar\theta_{tn}}{N_t},$$
and
$$\hat\kappa_t = \frac{\sum_{n=1}^{N_t} (\bar\theta_{tn} - \bar \theta_t)^2}{N_t}.$$

\section{Application to Continuous Testing}
{We provide a further discussion on applying the proposed method to continuous testing under any designs (e.g.,
CAT, MST, and fixed-form testing). We let each time point $t$ correspond to a
fixed period of time (e.g., one day). The duration of the time period may be chosen based on the test volume to allow adequate sample sizes in computing the monitoring statistics.
Let
$S_t$ be item pool for this period 
and $S^*_t \subset S_t$  be the set of items administered during this period.
It is possible that all the items in the pool are used during the period, and thus $S_t^* = S_t$. For adaptive testing such as CAT and MST, different test takers receive different sets of items in $S_t^*$. For each item $k$ in $S_t^*$, we will construct a monitoring statistic $X_{kt}$ based on item response data collected during this period and update the posterior probability of   $\tau_k < t$. For items not used in this test administration, its posterior probability will remain unchanged. Given the posterior probabilities, the local FNR risk measure is well-defined and the proposed procedure can be applied.

From the above discussion, we see that the proposed method can be applied to continuous testing, as long as we can construct a sensible monitoring statistic $X_{kt}$ and know its pre- and post-distributions exactly or partially and have exact or partial information about the prior distributions for change points,
so that Algorithms~1 or 2 can be applied. In fact, the SIR statistic defined in Section~4 can be adapted to this more general setting. We let $A_{kt}$ be the subset of test takers who have been assigned item $k$ in the $t$th test administration. Then the percent correct statistic $\bar Y_{kt}$ for item $k$ at time $t$ becomes
$\bar Y_{kt} = (\sum_{n \in A_{kt}}  Y_{ktn})/|A_{kt}|$. Based on this new percent correct statistic, an SIR statistic can be defined,  and its pre- and post-distributions can be approximated  under an IRT framework similarly as in Section~4.
Information about the prior distributions of change points may be obtained based on historical data.
}

\bibliographystyle{apalike}
\bibliography{ref}

\end{document}